\documentclass[11pt,preprint]{article}

\usepackage{amsmath, latexsym, amsfonts, amssymb, amsthm}
\usepackage{graphicx, color,hyperref,dsfont,epsfig,caption,wrapfig,subfig}
\usepackage{pstricks}
\usepackage{movie15}
\hypersetup{
    linktoc=page,
    linkcolor=red,          % color of internal links
    citecolor=blue,        % color of links to bibliography
    filecolor=blue,      % color of file links
    urlcolor=cyan,
    colorlinks=true           % color of external links
}

\numberwithin{equation}{section}

\newtheorem{theorem}{Theorem}[section]

\newtheorem{proposition}[theorem]{Proposition}
\newtheorem{corollary}[theorem]{Corollary}

\newtheorem{hypothesis}{Hypothesis}
\newcounter{conj}
\newtheorem{conjecture}[conj]{Conjecture}

\theoremstyle{definition}

 \title{Liouville quantum gravity on the annulus}

 \author{ Guillaume Remy\footnote{ D\'epartement de math\'ematiques et applications, \'Ecole normale sup\'erieure, CNRS, PSL Research University, 75005 Paris, France. Research supported in part by the ANR grant Liouville (ANR-15-CE40-0013).}  }
  
  \date{\vspace{-5ex}}
  
\begin{document}

  \maketitle

   \begin{abstract}
In this work we construct Liouville quantum gravity on an annulus in the complex plane. This construction is aimed at providing a rigorous mathematical framework to the work of theoretical physicists initiated by Polyakov in 1981 \cite{Pol}. It is also a very important example of a conformal field theory (CFT). Results have already been obtained on the Riemann sphere \cite{Sphere} and on the unit disk \cite{Disk} so this paper will follow the same approach. The case of the annulus contains two difficulties: it is a surface with two boundaries and it has a non-trivial moduli space. We recover the Weyl anomaly - a formula verified by all CFT - and deduce from it the KPZ formula. We also show that the full partition function of Liouville quantum gravity integrated over the moduli space is finite. This allows us to give the joint law of the Liouville measures and of the random modulus and to write the conjectured link with random planar maps.
  \end{abstract}
   

  %\noindent{\bf MSC 2000 subject classifications:  60D05, 81T40,  81T20.}     
    
  %\vspace{1cm}

  %\newpage
  %%\vspace{0.1cm}
  %%
  \tableofcontents

  \section{Introduction}

  The goal of this work is to provide a rigorous probabilistic construction of the theory of Liouville quantum gravity (LQG) on a surface with the topology of an annulus. This theory was first introduced by the physicist A. Polyakov in 1981 in his seminal paper ``Quantum Geometry of Bosonic Strings", see \cite{Pol}. The case of surfaces with boundary was studied in \cite{Alvarez} and the precise case of the annulus was studied in \cite{Martinec}. The probabilistic framework used throughout this paper was introduced in \cite{Sphere} where the authors provide a construction of LQG for the Riemann sphere. Following the same approach, the theory has been defined on the unit disk in \cite{Disk}, on the complex tori in \cite{Tori} and on compact Riemann surfaces of higher genus in \cite{Genus}. The Riemann sphere is the simplest case as it corresponds to a simply connected compact surface without boundary. When considering the unit disk, technical difficulties appear due to the presence of the boundary and extra boundary terms have to be added in the Liouville action. The torus and higher genus surfaces have a non-trivial moduli space, meaning that all tori are not equivalent under conformal maps. In this case defining LQG requires to integrate over the moduli space with an appropriate measure. The annulus possesses two boundaries and has a non-trivial moduli space so we will encounter both of these difficulties.

The first building block of Liouville quantum gravity is the Liouville quantum field theory (LQFT), which in probabilistic terms corresponds to giving the law of a random field $\phi$ on an annulus $\Omega$ of radii $1$ and $\tau$.\footnote{$\tau \in (1, + \infty)$ parametrizes the space of non-conformally equivalent annuli. To define LQFT we work at fixed $\tau$ and then to define LQG we will perform an integration over $\tau$.} To make the construction clearer, let us make an analogy with Brownian motion. Physicists often define the law of Brownian motion using the Feynman path integral representation. Informally, let  $\Sigma_1$ be the space of paths $\sigma : [0,1] \rightarrow \mathbb{R} $ that start from $\sigma(0) = 0$. We can define the following functional on  $\Sigma_1$ by:
\begin{equation}
\forall\sigma\in\Sigma_1,\: S_{BM}(\sigma) = \frac{1}{2} \int_0^1|\dot{\sigma}(t)|^2 dt. 
\end{equation}  
   
This functional is minimal for $ \sigma \equiv 0$, this is the ``classical" solution. The quantum theory corresponds to considering all paths of $\Sigma_1$ with a probability density given by the exponential of $-S_{BM}$: $S_{BM}$ is the energy of the path, $e^{-S_{BM}}$ is the corresponding Boltzmann weight. This leads to the formal path integral definition of Brownian motion, for all suitable functionals $F$,
\begin{equation}
 \mathbb{E}[F((B_s)_{0 \leq s\leq 1})] = \frac{1}{\mathcal{Z}} \int_{\Sigma_1} F(\sigma) e^{-S_{BM}(\sigma)} D\sigma,
\end{equation} 
  where $D\sigma$ represents a formal uniform measure on $\Sigma_1$ and $\mathcal{Z}$ is a normalization constant. Although $D\sigma$ and $S_{BM}(\sigma)$ are ill-defined (a typical path $\sigma$ will not be differentiable), it is possible to give a rigorous meaning to this expression by subdividing $[0,1]$ into $n$ points and taking $n$ to infinity. The fact that we recover the law of Brownian motion is due to the Donsker theorem.
  
 Brownian motion is often seen as the canonical uniform random path in $\mathbb{R}^d$: it is the scaling limit of the simple random walk on the regular Euclidean lattice and it appears in a wide variety of problems in probability. A natural question is to ask what would be the equivalent in two dimensions, what is the canonical random geometry of a surface of given topology? More precisely, what is the canonical random Riemannian metric on a given surface? The answer is given by the Liouville quantum field theory. Just like for Brownian motion, there are two possible approaches to LQFT: one continuum approach that uses the Feynman integral formalism and one that uses the scaling limit of discrete models called random planar maps. We will focus on the first approach and construct directly the continuous object. For a given background Riemannian metric $g$ on our annulus $\Omega$, we want to consider the formal random metric $e^{\gamma \phi} g$ where $\gamma$ is a positive real number and where $\phi$ is the Liouville field whose law is formally given for all suitable functionals $F$ by,
 \begin{equation}\label{Path_int}
  \mathbb{E}[F(\phi)] = \frac{1}{\mathcal{Z}} \int_{\Sigma} F(X) e^{-S_L(X,g)} D_gX, 
 \end{equation}
with $D_gX$ being a formal uniform measure on the space $\Sigma$ of maps $X : \Omega \rightarrow \mathbb{R} $. $S_L(X,g)$ is the Liouville action defined for each map $X$ of $ \Sigma $ and each metric $g$ on $\Omega$ by,
 \begin{equation}
  S_L(X,g) = \frac{1}{4\pi} \int_{\Omega} ( \vert\partial^g X\vert^2 + Q R_g X + 4 \pi \mu e^{\gamma X}) d \lambda_g 
  + \frac{1}{2\pi} \int_{\partial \Omega} (Q K_g X + 2 \pi \mu_{\partial} e^{\frac{\gamma}{2} X}) d \lambda_{\partial g}, 
 \end{equation}
where $\partial_g$, $R_g$, $K_g$, $d \lambda_g$, and $d \lambda_{\partial g}$ stand respectively for the gradient, Ricci scalar curvature, geodesic curvature (along the boundary), volume form and line element along $\partial\Omega$ in the metric $g$. The parameters $\mu$, $\mu_{\partial} \geq 0$ (with $ \mu + \mu_{\partial} >0$) are respectively the bulk and boundary cosmological constants and $Q,\gamma$ satisfy $\gamma \in (0,2)$ and $ Q = \frac{2}{\gamma} + \frac{\gamma}{2}$. We notice that there are two boundary terms in our action which would not be present if $\Omega$ had no boundary. 
  
 One may wonder why the Liouville action is the correct action to define canonical random metrics. A first answer comes from physics, in particular from Polyakov in \cite{Pol} (see the appendix for some heuristics). An easier answer comes from the study of classical Liouville theory, meaning that we look for the functions $X$ minimizing the Liouville action. It is a well known fact of classical geometry that such a minimum $X_{min}$ is unique if it exists and the new metric $g' = e^{\gamma X_{min}} g $ on $\Omega$ is of constant negative curvature provided that $ Q_c = \frac{2}{\gamma}$\footnote{Here we write $Q_c = \frac{2}{\gamma}$ as this is the correct value in the classical theory where the goal is to minimize the Liouville action. In the quantum theory we will always have $Q =\frac{2}{\gamma} + \frac{\gamma}{2}$.}. In other words, the minimum of the Liouville action uniformizes the surface $(\Omega, g)$ and it is therefore natural to look at quantum fluctuations of the uniformized metric $e^{\gamma X_{min}} g $. This is precisely the meaning of \eqref{Path_int}.
 
As we have written it the path integral \eqref{Path_int} diverges for any surface of genus 0 or 1, including therefore the case of the annulus. To see this we can write the Gauss-Bonnet formula, given here for a boundaryless surface $M$ of genus $h$: $\int_M R_g d \lambda_g  = 8 \pi (1 -h)$. When $ h = 0$ or $1$, this implies that it is impossible to define on the surface a metric of constant negative curvature, meaning that $S_L(X,g)$ will have no minimum and therefore the path integral \eqref{Path_int} diverges. To solve this problem we proceed as in $\cite{Disk}$ and add insertion points. We consider the new expression,
\begin{equation}\label{Path_int2}
\mathbb{E}[F(\phi)] = \frac{1}{\mathcal{Z}} \int_{\Sigma} F(X) e^{\sum_{i=1}^n \alpha_i X(z_i) + \frac{1}{2}\sum_{j=1}^{n'} \beta_j X(s_j)} e^{-S_L(X,g)} D_gX,
\end{equation}
  where we have chosen $n$ insertion points $z_i \in \Omega$ with weights $\alpha_i \in \mathbb{R}$ in the interior of the annulus and $n'$ insertions points $s_j \in \partial \Omega $ with weights $\beta_j \in \mathbb{R}$ on the boundaries of the annulus ($n, n' \in \mathbb{N}$). We show that the following conditions known as the Seiberg bounds must be satisfied in order for \eqref{Path_int2}  to exist:
  \begin{align}
  \sum_{i=1}^n &\alpha_i + \sum_{j=1}^{n'} \frac{\beta_j}{2} > 0 \quad \text{and} \quad \forall i, \: \alpha_i < Q  , \quad    \forall j, \: \beta_j < Q.  
  \end{align}
Theses bounds imply that the minimum number of insertion points is one point in $\Omega$ or one point on the boundary $ \partial \Omega$. Fixing a point inside $\Omega$ is actually a stronger constraint than fixing a boundary point. Also note that choosing one boundary point is precisely the requirement to entirely determine a conformal automorphism of the annulus. From a geometric standpoint, we can also view insertion points as conical singularities of the metric which allow hyperbolic metrics to be defined on the surface.

Let us now outline the structure of our paper. We start in section \ref{GFFGMC} by introducing the Gaussian free field and the associated Gaussian multiplicative chaos measures and state the properties of these objects that we will need in the sequel. Then in section \ref{LQFT} we give a mathematical definition for the partition function of LQFT formally given by \eqref{Path_int2}. We establish some of its properties well known to physicists such as the Weyl anomaly (behavior under change of metric) and the KPZ formula (behavior under conformal automorphism). In subsection \ref{liouville_field} we define the Liouville field and the Liouville bulk and boundary measures for a fixed value of $\tau$. Lastly in section \ref{LQG} we construct the full theory of LQG by integrating over $\tau \in (1, + \infty)$ the partition function of LQFT along with the partition function of a matter field.\footnote{All partition functions will first be computed on an annulus of radii $1$ and $\tau$ and then we will perform the integration over $\tau$.} We prove that this integral over moduli space is convergent and thus we show that the partition function of LQG is well defined. We then give the joint law of the Liouville measures and of the random modulus and write the conjectured link with the scaling limit of random planar maps.

Our paper also clarifies a few points that were not addressed in \cite{Sphere, Disk}. First, we justify for any background metric $g$ the expression of the covariance of our field $X_g$ by diagonalizing an operator $T$ defined using Neumann boundary conditions. We also reconcile two possible approaches to Liouville conformal field theory - the one proposed in \cite{Sphere,Disk} and the framework of Gawedzki \cite{gaw} - by explaining that these two approaches correspond to different ways of introducing the metric dependence of the theory. This allows to provided a very simple proof of the KPZ formula which becomes a direct consequence of the Weyl anomaly. \\
  
\noindent

\emph{Acknowledgements}: I would like to thank Juhan Aru, Yichao Huang, R\'emi Rhodes, Xin Sun, Vincent Vargas and Tunan Zhu for all their comments that helped me improve this paper.

\section{Gaussian free field and Gaussian multiplicative chaos}\label{GFFGMC}

\subsection{Geometric background}

The first tool we need is a description of the space of metrics on a surface with the topology of an annulus. The uniformization theorem tells us that any metric $g$ can be written up to a change of coordinates $e^{\varphi} dx^2$ where $e^{\varphi}$ is the Weyl factor and $dx^2$ is the Euclidean metric on an annulus $\Omega = \{ z , 1 < |z| < \tau \} $ with radii $1$ and $\tau$, the value of $\tau$ being uniquely determined. The parameter $\tau \in (1, +\infty) $ parametrizes the space of non conformally equivalent annuli known as the moduli space. We will always work at fixed $\tau$ with the exception of in section \ref{LQG} where we perform an integration over $\tau$.

For the annulus given by $\Omega = \{ z , 1 < |z| < \tau \} $, we consider the inner and outer boundaries $\partial \Omega_1 = \{ z , |z| = 1 \}$ and $\partial \Omega_{\tau} = \{ z , |z| = \tau \}$ and we set $\partial \Omega = \partial \Omega_1 \cup \partial \Omega_{ \tau }$. When $\Omega$ and $\partial \Omega$ are equipped with the Euclidean metric, we will write $ d \lambda$ and $ d \lambda_{\partial}$ for the Lebesgue measures on $\Omega$ and $\partial \Omega$,  $\Delta$ for the Laplace-Beltrami operator and $\partial_n$ for the Neumann operator. The convention for $\partial_n$ is that the derivative is computed along the vector $\overrightarrow{n}$ normal to the boundary and pointing outward of $\Omega$. For a general metric $g$ we will use the notations $ \partial^{g}$, $\Delta_g$, $\partial_{n_g}$, $R_g$, $K_g$, $d \lambda_g$, and $d \lambda_{\partial g}$ for the gradient, Laplace-Beltrami operator, Neumann operator, Ricci scalar curvature, geodesic curvature, volume form and line element along $\partial\Omega$ in the metric $g$. Under a Weyl rescaling of the metric $g' =e^{\varphi} g$, these objects behave in the following way: $\Delta_{g'} = e^{-\varphi} \Delta_{g}$, $\partial_{n_{g'}} = e^{-\varphi/2} \partial_{n_{g}}$, $ d \lambda_{ g'} =  e^{\varphi}  d \lambda_{ g}$, and $d \lambda_{ \partial g'} =  e^{\varphi/2}  d \lambda_{ \partial g}$. We also have:
\begin{align}
&R_{g'} = e^{-\varphi} (R_g - \Delta_g \varphi) \label{change_R} \\
&K_{g'} = e^{-\varphi/2} (K_g + \partial_{n_g} \varphi/2 )   \label{change_K}
\end{align}

These formulas will be useful for $g = e^{\varphi} dx^2$.
In the case of the Euclidean metric $dx^2$ we have $R=0$ on $ \Omega $, $K= -1$ on $\partial \Omega_1 $, and $K= \frac{1}{\tau}$ on $\partial \Omega_{ \tau } $. Let us recall the Gauss-Bonnet theorem,
\begin{equation}\label{GB}
\int_{\Omega} R_g d \lambda_g + 2 \int_{\partial \Omega} K_g d \lambda_{\partial g} = 4 \pi \chi(\Omega) = 0,
\end{equation}
where in the case of the annulus the Euler characteristics $\chi(\Omega)$ is worth $0$. Similarly, the Green-Riemann formula gives:
\begin{equation}\label{GR}
\int_{\Omega} \psi \: \Delta_g \varphi \: d \lambda_g + \int_{\Omega} \partial^g \varphi \cdot \partial^g \psi \: d \lambda_{g} = \int_{\partial \Omega} \partial_{n_g} \varphi \: \psi \: d \lambda_{\partial g}. 
\end{equation}
Here again the derivative $\partial_{n_g}$ is along the normal vector pointing outwards of $\Omega$. For a function $f$ defined on $\Omega$ or $\partial \Omega$, we define the averages $m_g(f)$ and $m_{\partial g}(f)$ with respect to $d \lambda_g$ and $d  \lambda_{\partial g}$ as $m_g(f) = \frac{1}{\lambda_g(\Omega)} \int_{\Omega} f d\lambda_g$ and $m_{\partial \Omega}(f) = \frac{1}{\lambda_{\partial g}(\partial \Omega)} \int_{\partial \Omega} f d  \lambda_{\partial g} $. $m(f)$ and $m_{\partial}(f)$ will correspond to the averages in the flat metric. Finally let us mention that the conformal automorphisms of our annulus $\Omega$ are the rotations plus the inversion that exchanges both boundaries. Therefore a conformal automorphism can be written either $ z \rightarrow e^{i\theta} z $ or $ z \rightarrow e^{i \theta} \frac{\tau}{z}$ for a $\theta \in [0, 2 \pi )$.

\subsection{Gaussian free field}\label{gradient_term}

The first step in the construction of LQFT is to interpret the formal density $ D_{g} X  e^{- \frac{1}{4 \pi} \int_{\Omega} \vert \partial^g X \vert^2 d \lambda_g}$ as the density of a Gaussian free field (GFF) with certain boundary conditions. Let $H_0^1(\Omega)$ be the Sobolev space on $\Omega$ with $ \forall f \in H_0^1(\Omega), \: \int_{\partial \Omega} f d \lambda_{\partial g} =0 $. $H_0^1(\Omega)$ is a Hilbert space with respect to the inner product given for $f,h \in H_0^1(\Omega) $ by:
\begin{equation}
\langle f, h \rangle  := \int_{\Omega} ( \partial f \cdot \partial h) d \lambda. 
\end{equation}
We define the operator $T$ on $H_0^1(\Omega)$ by, $\forall f \in H_0^1(\Omega), T(f) = u$ where the function $u$ is the unique solution to the following Neumann problem:
\begin{equation}\label{neu_prob}  
  \left \lbrace  \begin{array}{lcl} \Delta_g u(z) = - 2 \pi f(z) & \text{for} & z\in \Omega  \\ 
 \partial_{n_g} u(z)  = -\frac{ c_g(z) }{ \tau +1 }    \int_{\Omega} f(z') d\lambda_g(z') & \text{for} &  z \in \partial \Omega    \\ 
 \int_{\partial \Omega} u(z) d \lambda_{\partial g}(z) = 0 & & \end{array}  \right. 
\end{equation}
The function $c_g(z)$ defined on $ \partial \Omega $ that appears in the above is equal to $1$ in the case of the Euclidean metric and has a general expression in any metric $g$ given by \eqref{def_cg}. We can give the expression of $u$ in terms of $f$ using the Green's function $G_g$:
\begin{equation}\label{neu_prob2}
u(z) = \int_{\Omega} G_g(z,z') f(z') d\lambda_g(z'). 
\end{equation}
The expression for $G_g$ will be given in subsection \ref{sub_green}. $T$ is an auto-adjoint compact operator of $H_0^1(\Omega)$, therefore we can find a basis $e_i, i \in \mathbb{N}^*$ of eigenvectors with associated eigenvalue $\lambda_i^{-1} >0 $. The $e_{i}$ are normalized so that $ \int_{\Omega} e_{i}(z)^{2} d \lambda_{g}(z) =1 $. With this basis of eigenfunctions we can formally write, 
\begin{equation}
 D_gX e^{- \frac{1}{4 \pi} \int_{\Omega} \vert \partial^g X \vert^2 d \lambda_g} = \prod_{i=1}^{+\infty} dx_i \: e^{-\frac{1}{2} \sum_{i=1}^{+\infty} \lambda_i x_i^2},
\end{equation}
where each $dx_i$ is a Lebesgue measure on $\mathbb{R}$. This corresponds to the density of a GFF $X_g$ which can be written,\footnote{Here the dependence in $g$ is contained in  $e_i$ and in $\lambda_i$.}
\begin{equation}\label{def_GFF}
X_g(z) = \sum_{i=1}^{+ \infty} \frac{x_i e_i(z)}{\sqrt{\lambda_i}},
\end{equation}
where in this expression the $x_i$ are i.i.d. standard Gaussian variables. This sum converges in the Sobolev space $H^{-1}(\Omega)$. Just like for all the $e_i$, we have $\int_{\partial \Omega} X_g \: d\lambda_{\partial g} = 0$. But this condition is arbitrary as the formal density $ D_gX e^{- \frac{1}{4 \pi} \int_{\Omega} \vert \partial^g X \vert^2 d \lambda_g}$ only defines $X_g$ up to a constant. To deal with this problem we will actually consider $X_g +c$, where $X_g$ is the GFF defined above with zero average on the boundary and $c$ is a constant integrated according to the Lebesgue measure on $\mathbb{R}$. We must also take into account the value of the formal normalization constant $\mathcal{Z}_{GFF}(g)$ of our Gaussian density,
\begin{equation}\label{partition_GFF_def}
\mathcal{Z}_{GFF}(g) = \int_{\Sigma'} D_gX e^{- \frac{1}{4 \pi} \int_{\Omega} \vert \partial^g X \vert^2 d \lambda_g},
\end{equation}
where here $\Sigma'$ stands for the space of maps $X : \Omega \mapsto \mathbb{R}$ with zero average on $\partial \Omega$. The dependence of $\mathcal{Z}_{GFF}(g)$ on the background metric $g$ is given by \eqref{shift_partition_GFF}. To summarize this construction we have for suitable functionals $ \hat{F}$:
\begin{equation}
\int_{\Sigma} D_g X e^{ -\frac{1}{4 \pi } \int \vert \partial^g X \vert^2 d  \lambda_g} \hat{F}(X) = \mathcal{Z}_{GFF}(g) \int_{\mathbb{R}} dc \: \mathbb{E} [ \hat{F}(X_g +c)].
\end{equation}
To define LQFT we will choose, 
\begin{equation}
\hat{F}(X) = F(X)e^{\sum_i \alpha_i X(z_i) + \frac{1}{2}\sum_j \beta_j X(s_j)} e^{ -\frac{1}{4 \pi} \int_{\Omega} ( Q R_g X + 4 \pi \mu e^{\gamma X}) d \lambda_g} e^{ -\frac{1}{2 \pi} \int_{\partial \Omega} ( Q K_g X + 2 \pi \mu_{\partial} e^{\frac{\gamma}{2} X}) d \lambda_{\partial g}}, 
\end{equation}
where $F$ is again some functional, the $(z_i, \alpha_i)$ are the bulk insertion points, the $(s_j, \beta_j)$ are the boundary insertion points, and the last part corresponds to the curvature and exponential terms of the Liouville action. Extra care will be required as $X_g$ lives in the space of distributions $H^{-1}(\Omega)$ so $e^{\gamma X_g}$ is ill-defined.

  \subsection{Properties of the Green's function}\label{sub_green}
We can check that the function $G_g(z,z')$, which solves problem \eqref{neu_prob} through the relation \eqref{neu_prob2}, is also the covariance of our Gaussian free field defined by \eqref{def_GFF}:
\begin{equation}
G_g(z,z') = \mathbb{E}[ X_g(z) X_g(z')] = \sum_{i=1}^{\infty} \frac{e_i(z) e_i(z')}{\lambda_i}.
\end{equation}
The Neumann boundary conditions of \eqref{neu_prob} entirely determine $G_g$. Other boundary conditions would also be possible but they lead to major inconsistencies in the theory\footnote{For instance for Dirichlet boundary conditions the Weyl anomaly of section \ref{sec_weyl} fails to hold.}. In the case of the Euclidean metric $g= dx^2$, we simply write $G$ for the Green's function. From the computation detailed in the appendix we obtain the following expression for $G$ using polar coordinates $z=r e^{i \theta}$, $z' = \rho e^{i \phi}$,
\begin{equation}\label{green}
 G(r,\theta,\rho, \phi) =   g_0(r,\rho) +2\sum_{n=1}^{\infty} {g}_n(r,\rho) \cos n(\theta - \phi) + \ln \frac{|\tau^4 z^2 z'^2|}{|1 - z \overline{z'}| |\tau^2 - z \overline{z'}||z - z'| |\tau^2 z - z'|},  
\end{equation}  
  where the factor $|\tau^2 z - z'|$ holds for $r < \rho$ and is replaced by $| z - \tau^2 z'|$ for $r > \rho$, and where
  \begin{align}
  &g_n(r,\rho) = \frac{r^{-n} \rho^{-n} }{2n \tau^{2n}(\tau^{2n} - 1)} \left\{ \begin{array}{lcl} (\tau^{2n} + \rho^{2n})(r^{2n} + 1), & \text{for}  & r \leq \rho \\ (\tau^{2n} + r^{2n})(\rho^{2n} + 1), & \text{for}  &  r \geq \rho \end{array} \right. \\
  &g_0(r,\rho) = \left\{ \begin{array}{lcl}  \frac{ \ln(r) + \tau^2 \ln (\tau/\rho) + \tau \ln (r/ \rho) }{(\tau +1)^2}, & \text{for} & r \leq \rho \\  \frac{ \ln(\rho) + \tau^2 \ln (\tau/r) +  \tau \ln ( \rho/r) }{(\tau+1)^2}, & \text{for} &  r \geq \rho \end{array} \right.
  \end{align}
We now give some useful properties of the Green's function. By construction we have:
\begin{align}
&G(z,z') = G(z',z) \\
&\Delta G(z,z') = - 2 \pi \delta_0(z-z')
\end{align}
  By direct computation on the expression of the Green's function, we have the following results:
\begin{align}\label{der_g}
  &\frac{\partial G(r,\theta,\rho, \phi)}{\partial r} |_{r = 1} = \frac{\partial g_0(r,\rho)}{\partial r} |_{r = 1} =  \left\{ \begin{array}{lcl} \frac{1}{\tau +1} , & \text{for} & \rho >1 \\  \frac{-\tau}{\tau +1}, & \text{for} &  \rho = 1 \end{array} \right.  \\
  &\frac{\partial G(r,\theta,\rho, \phi)}{\partial r} |_{r = \tau} = \frac{\partial g_0(r,\rho)}{\partial r} |_{r = \tau} = \left\{ \begin{array}{lcl} \frac{-1}{\tau +1} , & \text{for} & \rho < \tau \\  \frac{1}{\tau(\tau +1)}, & \text{for} &  \rho = \tau \end{array} \right. \nonumber
\end{align}
Concerning the integral of the Green's function over the boundaries, we get:
\begin{equation}\label{bou_g} 
\int_{\partial \Omega} G(z, z') d\lambda_{\partial}(z)= \frac{2 \pi  }{(\tau +1)^2} (\tau^2 \ln(\tau/\rho) + \tau \ln(1/\rho ))+ \frac{2 \pi \tau }{(\tau +1)^2} ( \ln(\rho) + \tau \ln(\rho/\tau )) =0.
  \end{equation}
 Finally, we have the following formula valid for all (smooth) functions $f$ on $\Omega$ and for all $z \in \Omega$:
  \begin{equation}\label{gr_g}
  \int_{\Omega} G(z,z') \Delta f(z') d\lambda(z') - \int_{\partial\Omega} G(z,z') \partial_nf(z') d\lambda_{\partial}(z')  = -2 \pi (f(z) - m_{\partial} (f) ). 
  \end{equation}
  Indeed, using \eqref{GR} and \eqref{der_g}
  $$ \int_{\Omega} G(z,z') \Delta f(z') d\lambda(z') = -\int_{\Omega} \partial G(z,z') \partial f(z') d\lambda(z')
  + \int_{\partial\Omega} G(z,z') \partial_nf(z') d\lambda_{\partial}(z') $$
  and
  \begin{align*}
  -\int_{\Omega} \partial G(z,z') \partial f(z') d\lambda(z') &= \int_{\Omega} \Delta G(z,z') f(z') d\lambda(z') 
  - \int_{\partial\Omega}  \partial_n G(z,z')f(z') d\lambda_{\partial}(z') \\
  &= -2\pi f(z)  + \frac{1}{\tau +1} \int_{\partial \Omega} f(z') d \lambda_{\partial}(z') \\
  &= -2\pi(f(z) - m_{\partial }(f)).
  \end{align*}
We now discuss the relation between $G_g$ and $G$ when $ g = e^{\varphi} dx^2 $. By definition $G$ solves our Neumann problem \eqref{neu_prob} in the flat metric $dx^2$. We will show that $G_g$ is given by the following expression,
\begin{equation}\label{green_g}
G_g(z,z') = G(z,z') - m_{ \partial g}( G(z,\cdot)) - m_{ \partial g}( G(\cdot,z')) + m_{ \partial g}( G(\cdot,\cdot)),\footnote{Here we have defined $m_{ \partial g}( G(\cdot,\cdot)) := \frac{1}{\lambda_{\partial g}(\partial \Omega)^2} \int_{\partial \Omega} \int_{\partial \Omega} G(z,z') d \lambda_{\partial g}(z) d \lambda_{\partial g}(z') $. }
\end{equation}
which means that we must check that this function solves our problem \eqref{neu_prob} in the metric $g$. By construction of $G_g$, we clearly have $\int_{\partial \Omega} u(z) d\lambda_{\partial g}(z) = 0$. Next we must check that $ \Delta_g u(z) = - 2 \pi f(z) $:
\begin{align*}
\Delta_g u(z) &= e^{- \varphi(z)} \int_{\Omega} \Delta G(z,z') f(z') d \lambda_g(z') - e^{- \varphi(z)} \int_{\Omega} \Delta (m_{\partial g} G(z, \cdot)) f(z') d \lambda_g(z') \\
 &= - 2 \pi f(z) - e^{-\varphi(z)} \frac{\lambda_g(\Omega)}{\lambda_{\partial g }(\partial \Omega)} m_g(f) \int_{\partial \Omega} \Delta G(z,z') d \lambda_{\partial g}(z')\\
 &= - 2 \pi f(z).
\end{align*}
Lastly we look at the normal derivative, the function $c_g(z)$ introduced in \eqref{neu_prob} will be chosen to match the following computation,
\begin{align}
\partial_{n_g} u(z) &= e^{-\varphi(z)/2} \int_{\Omega} \partial_n G(z,z') f(z') d \lambda_g(z') - e^{-\varphi(z)/2} \int_{\Omega} \partial_n m_{\partial g} G(z, \cdot) f(z') d \lambda_g (z')  \nonumber \\
&= - \frac{e^{-\varphi(z)/2}}{\tau +1} \int_{\Omega}  f(z') d \lambda_g(z') - e^{-\varphi(z)/2} \partial_n m_{\partial g} G(z, \cdot) \int_{\Omega}  f(z') d \lambda_g(z') \nonumber \\
& = -\frac{c_g(z)}{\tau +1} \int_{\Omega} f(z') d \lambda_g(z'),
\end{align}
where a straightforward computation shows that $c_g(z)$ is given by:
\begin{equation}\label{def_cg}
c_g(z)  = \left\{ \begin{array}{lcl} e^{-\varphi(z)/2}(1 + \tau  \lambda_{\partial g}(\partial \Omega_1) -    \lambda_{\partial g}(\partial \Omega_{\tau}) ) , & \text{for} & z \in \partial \Omega_1 \\  e^{-\varphi(z)/2}(1 -   \lambda_{\partial g}(\partial \Omega_1) + \frac{1}{\tau}   \lambda_{\partial g}(\partial \Omega_{\tau}) ), & \text{for} &  z \in \partial \Omega_{\tau} \end{array} \right.
\end{equation}
We notice that $c_g(z) =1$ for the Euclidean metric. From all the above we have shown that \eqref{green_g} gives the correct expression for $G_g$. From this we can easily deduce
\begin{equation}
X_g(z) \overset{law}{=} X(z) - m_{\partial g}(X)
\end{equation}
where we have written $X$ for the GFF in the flat metric of covariance $G$. In particular this tells us that $X_g + c$ for $g = e^{\varphi} dx^2$ and $c$ distributed according to the Lebesgue measure on $\mathbb{R}$ is independent of $\varphi$ because we can simply make a shift in the integral over $c$ to remove the constant $m_{\partial g}(X)$. More precisely, for $g = e^{\varphi} dx^2$ and any suitable functional $F$:
\begin{equation}
\int_{\mathbb{R}} dc \; \mathbb{E} [ F(X_{g} +c)] = \int_{\mathbb{R}}  dc \; \mathbb{E} [ F(X +c)].
\end{equation}

  \subsection{Circle average regularization}\label{sec_circle}
    
Let $X$ be the GFF on $\Omega$ with covariance function given by $G$. Because $X$ lives almost surely in the space of distributions $H^{-1}(\Omega)$ and we wish to define the exponential of $X$,  we need to introduce a regularization procedure. We will choose a circle average regularization, more precisely for $\epsilon >0$ we call $l_{\epsilon}(x)$ the length of the arc 
    $ A_{\epsilon}(x) = \{z \in \Omega; |z-x| = \epsilon\}$ and we set:
    \begin{equation}
     X_{g,\epsilon}(x) = \frac{1}{l_{\epsilon}(x)} \int_{A_{\epsilon}(x)} X(x + g(x)^{-1/2} s) ds.
    \end{equation}
For a point $x$ at a distance larger than $\epsilon$ from the boundary this definition gives the standard circle average. The term $g^{-1/2}$ is added because the regularization must depend on the background metric $g$ in order to get a consistent theory.\footnote{This way of introducing the metric dependence is slightly different than the one of \cite{Sphere, Disk}. The advantage of our method is that it allows us to recover exactly the framework of \cite{gaw}, see Theorem \ref{WEYL}. The KPZ formula of section \ref{sec_KPZ} then becomes a direct consequence of the Weyl anomaly of section \ref{sec_weyl}.} Here $g(x)= e^{\varphi(x)}$ with as usual $g = e^{\varphi}dx^2 $.  We have the following result:
    
\begin{proposition}\label{circlegreen}    
As $\epsilon \rightarrow 0$ we have the following convergences.    \\
1) Uniformly over all compact subsets of $\Omega$: 
    $$ \mathbb{E}[X_{g , \epsilon}(x)^2] + \ln \epsilon \underset{\epsilon \rightarrow 0}{\rightarrow} \frac{1}{2} \ln g(x) +  \ln g_P(x) + h(x)$$     
\\      
2)  Uniformly over $\partial \Omega$: 
    $$ \mathbb{E}[X_{g , \epsilon}(x)^2] + 2 \ln \epsilon \underset{\epsilon \rightarrow 0}{\rightarrow} \ln g(x) +  h_{\partial}(x)$$ 
Here $g$ is our background metric, $g_P(x) := \frac{1}{|1 -  |x|^2| |\tau^2 - |x|^2|} $ is the term that diverges on the boundaries, and $h$, $h_{\partial}$ are continuous functions independent of $g$ whose explicit expressions are given in the proof below.
\end{proposition}

\begin{proof}
This proposition follows from an explicit computation with the Green's function of $X$. For a point $x \in \Omega$, we have:
\begin{align*}
\mathbb{E}[X_{g , \epsilon}(x)^2] &= \mathbb{E} [ \frac{1}{4\pi^2} \int_{0}^{2\pi} X(x + \epsilon g(x)^{-1/2}  e^{i \theta}) d\theta \int_0^{2 \pi} X(x + \epsilon g(x)^{-1/2} e^{i \theta'}) d\theta'  ] \\
 &=  \frac{1}{4\pi^2} \int_{0}^{2\pi}   \int_{ 0}^{2\pi} G(x + \epsilon g(x)^{-1/2} e^{i \theta},x + \epsilon g(x)^{-1/2} e^{i \theta'} ) d\theta d\theta'
\end{align*}
    We then notice that in the expression of the Green's function, only the term $\ln \frac{1}{|z-z'|}$ diverges when $\epsilon$ goes to 0. We have added consequently the term $\ln \epsilon$ to cancel this divergence. From this computation we get the expression of the function $h$,
\begin{align*}
h(x)  &= g_0(r,r) + 2 \sum_{n=1}^{\infty} g_n(r,r) + \ln \frac{\tau^4 r^3}{\vert \tau^2 - 1 \vert} +  \frac{1}{4\pi^2} \int_{0}^{2\pi} \int_{0}^{2\pi} d\theta d\theta' \ln \frac{1}{\vert e^{i\theta} - e^{i\theta'} \vert }   \\
&= g_0(r,r) + 2 \sum_{n=1}^{\infty} g_n(r,r) + \ln \frac{\tau^4 r^3}{\vert \tau^2 - 1 \vert},
\end{align*}   
where $r = \vert x \vert$. Similarly for the boundary $\partial \Omega$ we have the following expression for $h_{\partial}$:
\begin{small}
\begin{align*}
&h_{\partial}(x) = \left\{ \begin{array}{lcl}  g_0(1,1) + 2\sum_{n=1}^{\infty} {g}_n(1,1) + \ln \frac{|\tau^4|}{|\tau^2-1|^2}   -  \frac{1}{\pi^2} \int_{0}^{\pi} \int_{0}^{\pi} \ln|(e^{i \theta} - e^{i \theta'} ) \frac{\overline{x} e^{i \theta} + x e^{-i \theta'}} {x}| d\theta  d \theta'  , & \text{for} & x \in \partial \Omega_1 \\  g_0(\tau,\tau) + 2\sum_{n=1}^{\infty} {g}_n(\tau,\tau) + \ln \frac{|\tau^6|}{|\tau^2-1|^2} -  \frac{1}{\pi^2} \int_{0}^{\pi} \int_{0}^{\pi} \ln| (e^{i \theta} - e^{i \theta'})   \frac{\overline{x} e^{i \theta} + x e^{-i \theta'}} {x}| d\theta  d \theta'
, & \text{for} &  x \in \partial \Omega_{\tau} \end{array} \right.
\end{align*}
\end{small}
\end{proof}
 
\subsection{Gaussian multiplicative chaos}\label{sec_GMC}
Now that our field $X$ is random distribution, we need to tackle the problem of giving sense to the terms $ \mu e^{\gamma X}$ and $\mu_{\partial} e^{ \frac{\gamma}{2} X}$ as the exponential of a distribution is ill-defined. This can be done using the theory Gaussian multiplicative chaos which was introduced by Kahane in \cite{Kah} and developed by others, see \cite{review, Houches}. The results from these papers allow us to claim: 
\begin{proposition}\label{def_GMC}
Let $\gamma \in (0,2)$ and let $g$ be a metric on $\Omega$. We introduce for $\epsilon >0$ the random measures:
\begin{align*}
&M_{\gamma, g, \epsilon}(dx) = \epsilon^{\frac{\gamma^2}{2}} e^{ \gamma X_{g, \epsilon} (x) } d \lambda_g(x), \\
&M^{\partial}_{\gamma, g, \epsilon}(dx) =\epsilon^{\frac{\gamma^2}{4}} e^{ \frac{\gamma}{2} X_{g, \epsilon} (x) } d \lambda_{\partial g}(x).
\end{align*}
We then define the random measures $ M_{\gamma, g}(dx) $ on $\Omega$ and $M^{\partial}_{\gamma, g}(dx) $ on $\partial \Omega$ as the following limits in probability
\begin{align*}
&M_{\gamma, g}(dx)  :=\lim_{ \epsilon \rightarrow 0} M_{\gamma, g, \epsilon}(dx) = \lim_{ \epsilon \rightarrow 0}  e^{\gamma X_{\epsilon}(x) - \frac{\gamma^2}{2} \mathbb{E}[ X_{\epsilon}(x)^2]} g(x)^{ \frac{\gamma^2}{4}} g_P(x)^{\frac{\gamma^2}{2}} e^{\frac{\gamma^2}{2}h(x)} d\lambda_{g}(x)  \\
&M^{\partial}_{\gamma, g}(dx) :=  \lim_{ \epsilon \rightarrow 0} M^{\partial}_{\gamma,  g, \epsilon}(dx)  = \lim_{ \epsilon \rightarrow 0} e^{\frac{\gamma}{2} X_{\epsilon}(x) - \frac{\gamma^2}{8} \mathbb{E}[ X_{\epsilon}(x)^2]} g(x)^{\frac{\gamma^2}{8}} e^{\frac{\gamma^2}{8}h_{\partial}(x)} d\lambda_{ \partial g}(x)
\end{align*}
in the sense of weak convergence of measures respectively over $\Omega$ and $ \partial \Omega$.
\end{proposition}   
In order to prove the Weyl anomaly we will need to write how these measures behave under a conformal change of metric. More precisely for $g' = e^{\varphi} g$ we have: 
\begin{align}\label{shift_GMC}
&M_{\gamma, g'}(dx) = e^{(1 + \frac{\gamma^2}{4}) \varphi(x) } M_{\gamma, g}(dx), \\
&M^{\partial}_{\gamma, g'}(dx) = e^{\frac{1}{2}(1 + \frac{\gamma^2}{4}) \varphi(x) } M^{\partial}_{\gamma, g}(dx). \nonumber
\end{align}
To have a finite partition function (see  section \ref{sec_def_p}), we must show that these measures give an almost surely finite mass to the annulus and to its boundary. Here we write the results in the flat metric $g =dx^2$. We will use the notations $M_{\gamma,dx^2}$ and $M^{\partial}_{\gamma,dx^2}$ for the bulk and boundary GMC measures in the flat metric. For the boundary measure the expectation of the total mass of $\partial \Omega$ is finite:
$$ \mathbb{E}[M^{\partial}_{\gamma, dx^2}(\partial \Omega) ] = \int_{\partial \Omega} e^{\frac{\gamma^2}{8} h_{\partial}(x) } d\lambda_{\partial}(x) < + \infty.  $$
On the other hand for the bulk measure on $\Omega$, this is not straightforward at all as for instance the expectation
is infinite as soon as $\gamma^2 \geq 2$ because of the divergence of $g_P(x)$ on the boundary:
$$  \mathbb{E} [M_{\gamma, dx^2}(\Omega)] = \int_{\Omega} g_P(x)^{\frac{\gamma^2}{2}} e^{\frac{\gamma^2}{2} h(x) } d\lambda(x) = + \infty \: \: \text{when} \: \:  \gamma^2 \geq 2. $$
But as it is shown in \cite{Disk}, the random variable $M_{\gamma, dx^2}(\Omega)$ is almost surely finite for all values of $\gamma \in (0,2)$. Therefore we have:
  
  \begin{proposition}\label{fin_mass}
For $\gamma \in (0,2)$ the following quantities are almost surely finite:
$$ M_{\gamma, g}(\Omega) \: \text{  and  }  \:  M^{\partial}_{\gamma, g}(\partial \Omega). $$
\end{proposition}

Finally we state here two very useful tools that we will need for sections \ref{LQFT} and \ref{LQG}. We start with the Girsanov transform (also called the Cameron-Martin formula) that we will always use in the following way:

\begin{proposition}\label{girsanov} (Girsanov transform)
Consider the Gaussian free field $X$ on $\Omega$ and let $f: \Omega \mapsto \mathbb{R}$ be a continuous function. We introduce the random variable $Y= \int_{\Omega} f(z) X(z) d \lambda(z) $. Then for any suitable functional $F$ we have:
$$ \mathbb{E}[ F((X(z))_{z \in \Omega}) e^{Y}] = e^{ \frac{\mathbb{E}[Y^2]}{2}} \mathbb{E}[ F(( X(z) + \mathbb{E}[X(z)Y] )_{z \in \Omega}) ]. $$
The same result holds if we define $Y$ with an integral over the boundary $\partial \Omega$.    
\end{proposition}

We will also need Kahane's convexity inequalities:
\begin{proposition}\label{kahane}(Convexity inequalities, see \cite{Kah})
Let $(Y(z))_{z\in \Omega}$ and $(Z(z))_{z\in \Omega}$ be continuous centered Gaussian fields such that
$$ \mathbb{E}[Y(z)Y(z')] \leq \mathbb{E}[Z(z) Z(z')].$$
Then for all convex (resp. concave) functions $F: \mathbb{R}_+ \rightarrow \mathbb{R}$ with at most polynomial growth at infinity
\begin{equation}\label{ineq}
 \mathbb{E}\left[F\left(\int_{\Omega} e^{Y(z) - \frac{\mathbb{E}[Y(z)^2]}{2}} d\lambda(z)   \right)\right] \leq (\text{resp.} \geq) \text{  } \mathbb{E}\left[F\left(\int_{\Omega} e^{Z(z) - \frac{\mathbb{E}[Z(z)^2]}{2}} d\lambda(z)   \right)\right].
\end{equation}
By using our regularization procedure we can apply this result to the case where $Y$ and $Z$ are Gaussian free fields. We also have the same result if we replace the integral over $\Omega$ by an integral over the boundary $\partial \Omega$.
\end{proposition}
  
\section{Liouville quantum field theory}\label{LQFT}

  \subsection{Defining the partition function}\label{sec_def_p}
We are now ready to give the expression of the partition function for a metric $g= e^{\varphi} dx^2$. To write the regularized partition function we must add the renormalization in $\epsilon$ of Proposition \ref{def_GMC}  whenever there is an exponential of the field $X_{g, \epsilon}$. More explicitly\footnote{The definition we give is slightly different than the one in \cite{Sphere} or \cite{Disk}. Instead of adding a factor $\frac{Q}{2} \ln g$ to the field $X$ we include the dependence in $g$ in the regularization procedure of section \ref{sec_circle}. This leads to the framework of \cite{gaw} but it is completely equivalent to the framework of \cite{Sphere, Disk}.}
  \begin{align}\label{partition} 
   &\Pi_{\gamma,\mu,  \mu_{\partial}} ^{ (z_i, \alpha_i)_i, (s_j, \beta_j)_j}  (\epsilon, g, F) =
  \mathcal{Z}_{GFF}(g) \int_{\mathbb{R}} \mathbb{E} [F( X_{g , \epsilon} + c) \prod_{i=1}^n  \epsilon^{\frac{\alpha_i^2}{2}}  e ^{ \alpha_i ( X_{g , \epsilon} + c )(z_i)}  \prod_{j=1}^{n'}  \epsilon ^{\frac{\beta_j^2}{4}}  e ^{ \frac{\beta_j}{2} ( X_{g , \epsilon} + c )(s_j)}\nonumber \\
  &\times 
  \exp (- \frac{Q}{4 \pi} \int_{\Omega} R_g(c + X_{g , \epsilon}) d \lambda _g 
  - \mu e^{\gamma c} M_{\gamma, g, \epsilon}(\Omega) - \frac{Q}{2 \pi} \int_{\partial \Omega} K_g(c + X_{g , \epsilon}) d \lambda _{\partial g} 
  - \mu_{\partial} e^{\frac{\gamma}{2} c}  M^{\partial}_{\gamma, g, \epsilon}( \partial \Omega) )] dc
  \end{align}
Our parameters $\gamma$, $Q$, $\mu$, $\mu_{\partial}$ satisfy $\gamma \in (0,2)$, $Q = \frac{\gamma}{2} + \frac{2}{\gamma}$, $\mu \geq 0$, $\mu_{\partial} \geq 0$, and $ \mu + \mu_{\partial} >0$. The term $\mathcal{Z}_{GFF}(g)$ is the partition function of the GFF coming from \eqref{partition_GFF_def}. It satisfies \eqref{weyl_boun}, see for instance \cite{dubedat}:
\begin{equation}\label{shift_partition_GFF}
 \mathcal{Z}_{GFF}(e^{\varphi} dx^2) = e^{ \frac{1}{96 \pi} (  \int_{\Omega} |\partial \varphi|^2 d\lambda + 4 \int_{\partial \Omega} K \varphi d \lambda_{\partial} ) } \mathcal{Z}_{GFF}(dx^2).
\end{equation}
The value of  $\mathcal{Z}_{GFF}(dx^2)$ has no importance for the construction of LQFT but it will play a role in the section on LQG as it depends on the outer radius $\tau$ of our annulus, see section \ref{sec_ghost}. The main goal of what follows is to prove that we can let $\epsilon$ go to $0$ and obtain a finite non zero limit for the partition function:
\begin{equation}
\Pi_{\gamma, \mu, \mu_{\partial}} ^{ (z_i, \alpha_i)_i, (s_j, \beta_j)_j}  (g, F) = 
  \lim_{\epsilon\rightarrow 0} \Pi_{\gamma, \mu , \mu_{\partial}} ^{ (z_i, \alpha_i)_i, (s_j, \beta_j)_j}  (\epsilon, g, F).
\end{equation}
  Existence and non-triviality of this limit will depend on the Seiberg bounds:
  \begin{align}
  \sum_{i=1}^n &\alpha_i + \sum_{j=1}^{n'} \frac{\beta_j}{2} > 0 \label{sei1}  \\
  &\forall i, \: \alpha_i < Q \label{sei2} \\
  &\forall j, \: \beta_j < Q \label{sei3}
  \end{align}
\begin{theorem}\label{th_existence}
We have the following alternatives:
 \begin{enumerate}
 \item If $\mu>0$ and $\mu_{\partial} > 0$, then the partition function converges and is non trivial if and only if \eqref{sei1} + \eqref{sei2} + \eqref{sei3} hold.
 \item If $\mu>0$ and $\mu_{\partial} = 0$ \textit{(resp. if $\mu=0$ and $\mu_{\partial} > 0$)}, then the partition function converges and is non trivial if and only if \eqref{sei1} + \eqref{sei2} hold \textit{(resp. if \eqref{sei1} + \eqref{sei3} hold)}.
 \item In all other cases, the partition function is worth $0$ or $+ \infty$.
 \end{enumerate}

\end{theorem}    

The arguments we will use to prove this theorem come from \cite{Sphere, Disk}. We will start by studying the flat metric, meaning that $R = 0$, $K = -1$  on $\partial \Omega_1$, $K = \frac{1}{\tau}$ on $ \partial \Omega_{\tau}$, and $ \ln g = 0$ (the generalization to all metrics is given by the Weyl anomaly, see section \ref{sec_weyl}). We are going to apply the Girsanov transform of Proposition \ref{girsanov} to the insertion points and to the curvature term. The curvature term shifts $X(z)$ by
\begin{equation}
- \frac{Q}{2 \pi} \int_{\partial \Omega} K(x) G(x,z) d \lambda_{\partial}(x) = -\frac{Q}{\tau +1 } ( \tau \ln(\vert z \vert /\tau) +  \ln\vert z \vert)  
\end{equation}
and we also get a global factor $ \exp(\frac{Q^2}{2} \ln \tau)$. For the insertions we use Proposition \ref{circlegreen} to write on $\Omega$,
  \begin{align}\label{pre_res2}
  \epsilon^{\frac{\alpha_i^2}{2}}  e ^{ \alpha_i  X_{ \epsilon} (z_i)}
   &=   e ^{ \alpha_i  X_{ \epsilon}(z_i) - \frac{\alpha_i^2}{2} \mathbb{E}[X_{ \epsilon}^2(z_i)]   }
       e ^{ \frac{\alpha_i^2}{2} \mathbb{E}[X_{ \epsilon}^2(z_i)] +\frac{\alpha_i^2}{2} \ln \epsilon } \\
   &= e ^{ \alpha_i  X_{ \epsilon}(z_i) - \frac{\alpha_i^2}{2} \mathbb{E}[X_{ \epsilon}^2(z_i)]   }
       e^{ \frac{\alpha_i^2}{2} ( \ln g_P(z_i) + h(z_i) ) } (1 + o(1)), \nonumber
\end{align}  
and on the boundary $\partial \Omega$,
\begin{align}\label{pre_res3}
  \epsilon^{\frac{\beta_j^2}{4}}  e ^{ \frac{\beta_j}{2}  X_{ \epsilon} (s_j)}
   &=   e ^{ \frac{\beta_j}{2}  X_{ \epsilon}(s_j) - \frac{\beta_j^2}{8} \mathbb{E}[X_{ \epsilon}^2(s_j)]   }
       e ^{ \frac{\beta_j^2}{8} \mathbb{E}[X_{ \epsilon}^2(s_j)] +\frac{\beta_j^2}{4} \ln \epsilon } \\
   &= e ^{ \frac{\beta_j}{2}  X_{ \epsilon}(s_j) - \frac{\beta_j^2}{8} \mathbb{E}[X_{ \epsilon}^2(s_j)]  }
     e^{ \frac{\beta_j^2}{8} h_{\partial}(s_j)   } (1 + o(1)), \nonumber
\end{align}
where in both cases the $o(1)$ is deterministic. From this we can easily apply the Girsanov transform to the insertions and get:
\begin{small}
\begin{align}\label{partition_reg}
 \Pi_{\gamma, \mu, \mu_{\partial}} ^{ (z_i, \alpha_i)_i, (s_j, \beta_j)_j}  (dx^2, F, \epsilon) =
     & \mathcal{Z}_{GFF}(dx^2) (\prod_{i=1}^n  g_P(z_i)^{\frac{\alpha_i^2}{2}} ) e^{ C(z,s)} \int_{\mathbb{R}} e^{(\sum_i \alpha_i + \sum_j \frac{\beta_j}{2}) c} \: \mathbb{E} [ F( X_{\epsilon} + H_{\epsilon} + c )  \\
   &\times  \exp( -\mu e^{\gamma c} \int_{\Omega} e^{\gamma H_{\epsilon}(x)}  M_{\gamma,dx^2,\epsilon}(dx) - \mu_{\partial} e^{\frac{\gamma c}{2}} \int_{\partial \Omega} e^{\frac{\gamma }{2}H_{\epsilon}(x)} M^{\partial}_{\gamma,dx^2,\epsilon}(dx)   )] dc. \nonumber
 \end{align} 
\end{small}
 This leads to the following limit for the partition function:
 \begin{small}
 \begin{align}\label{partition_2}
 \Pi_{\gamma, \mu, \mu_{\partial}} ^{ (z_i, \alpha_i)_i, (s_j, \beta_j)_j}  (dx^2, F) =
     &\mathcal{Z}_{GFF}(dx^2) (\prod_{i=1}^n  g_P(z_i)^{\frac{\alpha_i^2}{2}} ) e^{ C(z,s)} \int_{\mathbb{R}} e^{(\sum_i \alpha_i + \sum_j \frac{\beta_j}{2}) c}\: \mathbb{E} [ F( X + H + c )    \\
   &\times  \exp( -\mu e^{\gamma c} \int_{\Omega} e^{\gamma H(x)}  M_{\gamma,dx^2}(dx) - \mu_{\partial} e^{\frac{\gamma c}{2}} \int_{\partial \Omega} e^{\frac{\gamma }{2}H(x)} M^{\partial}_{\gamma,dx^2}(dx)   )] dc.  \nonumber
 \end{align} 
 \end{small}
In the above expressions we have introduced the notations:
\begin{align}
H(x) &= \sum_{i=1}^n \alpha_i G(x,z_i) + \sum_{j=1}^{n'} \frac{\beta_j}{2} G(x,s_j) - Q \ln \vert x \vert  \\ 
H_{\epsilon}(x) &= \frac{1}{l_{\epsilon}(x)} \int_{A_{\epsilon}(x)} H(x + s) ds \\ 
 C(z,s) &= \sum_{i<i'} \alpha_i \alpha_{i'} G(z_i,z_{i'}) + \sum_{j<j'} \frac{\beta_j \beta_{j'}}{4} G(s_j,s_{j'})
  + \sum_{i,j} \frac{\alpha_i \beta_{j}}{2} G(z_i,s_{j}) + \sum_{i=1}^n \frac{\alpha_i^2}{2} h(z_i) \\ 
  &+ \sum_{j=1}^{n'} \frac{\beta_j^2}{8} h_{\partial}(s_j)  
  + \frac{Q^2}{2} \ln \tau  - \sum_{i=1}^n Q \alpha_i \ln \vert z_i \vert - \sum_{j=1}^{n'} \frac{Q \beta_j}{2} \ln \vert s_j \vert  \nonumber
\end{align}
Let us now explain why the Seiberg bounds \eqref{sei1}-\eqref{sei3} are required to get a non-diverging non-zero limit for \eqref{partition_reg}. These results are found again in \cite{Sphere, Disk}.

The first inequality $ \sum_i \alpha_i + \sum_j \frac{\beta_j}{2} > 0 $ controls the $c \rightarrow - \infty $ divergence of the integral. It is required for the regularized partition function \eqref{partition_reg} to exist. Let:
\begin{align*}
&Z_{0,\epsilon} = \epsilon^{\frac{\gamma^2}{2}} \int_{\Omega} e^{\gamma(X_{\epsilon} + H_{\epsilon})} d\lambda, \\
&Z_{0, \epsilon}^{\partial} = \epsilon^{\frac{\gamma^2}{4}} \int_{\partial  \Omega} e^{\frac{\gamma}{2}(X_{\epsilon} + H_{\epsilon})} d\lambda_{\partial}.
\end{align*}
We see that $|H_{\epsilon}| \leq C_{\epsilon}$ where $C_{\epsilon}$ is a constant depending only on $\epsilon$. Hence we know that $\mathbb{E} [Z_{0, \epsilon}] < \infty$ and thus $Z_{0, \epsilon} < \infty $ almost surely. The same thing holds for $Z_{0,\epsilon}^{\partial}$. Therefore, we can find $A>0$ such that $\mathbb{P}(Z_{0, \epsilon} \leq A , Z_{0, \epsilon}^{\partial} \leq A) >0$ which gives
  $$ \Pi_{\gamma, \mu, \mu_{\partial}} ^{ (z_i, \alpha_i)_i, (s_j, \beta_j)_j}  (dx^2, 1, \epsilon) \geq  \tilde{C} \int_{-\infty}^0 e^{(\sum_i \alpha_i + \sum_j \frac {\beta_j}{2})c} e^{ - \mu e^{\gamma c} A - \mu_{\partial} e^{\frac{\gamma}{2} c} A  } \: \mathbb{P}(Z_{0, \epsilon} \leq A , Z_{0, \epsilon}^{\partial} \leq A) dc $$
for some constant $\tilde{C} >0$. This last integral diverges if the condition $ \sum_i \alpha_i + \sum_j \frac{\beta_j}{2} > 0 $ fails to hold. This gives the bounds \eqref{sei1}.
 
The next thing we need to check is that the following integrals are finite a.s. (otherwise the partition function will be worth $0$):
$$ \int_{\Omega} e^{ \gamma H(x) } M_{\gamma, dx^2}(dx) \: \: \text{and} \: \int_{\partial \Omega} e^{  \frac{\gamma}{2} H(x) } M^{\partial}_{\gamma, dx^2}(dx).  $$
Proposition \ref{fin_mass} tells us that the integrals without $H$ are finite a.s. This means we can restrict ourselves to looking at what happens around the insertions: this will give the conditions \eqref{sei2} and \eqref{sei3}. Indeed for a bulk insertion $(z_i, \alpha_i)$, $e^{\gamma H(z)}$ behaves like $\frac{1}{\vert z -z_i\vert^{\alpha_i \gamma}}$ around $z_i$. For $B(z_i,r)$ a small ball around $z_i$ of radius $r>0$, the results of \cite{Sphere} tell us that:
 $$\int_ {B(z_i,r)}   \frac{1}{\vert x -z_i\vert^{\alpha_i \gamma}} M_{\gamma,dx^2}(dx) < + \infty \: \: a.s. \: \Longleftrightarrow \: \alpha_i < Q.$$
In the same way we can look at a boundary insertion $(s_j, \beta_j)$ and we have following \cite{Disk}:
$$\int_ {B(s_j,r) \cap \Omega }   \frac{1}{\vert x -s_j\vert^{\beta_j \gamma}} M_{\gamma,dx^2}(dx) < + \infty \, \, \text{and} \, \,  \int_ {B(s_j,r) \cap \partial \Omega  }  \frac{1}{\vert x -s_j\vert^{\frac{\beta_j \gamma}{2}}} M^{\partial}_{\gamma,dx^2}(dx) < + \infty  \, \,   \text{a.s.}  \,  \Longleftrightarrow \: \beta_j < Q.$$
Therefore we have all the conditions of Theorem \ref{th_existence}.

  \subsection{Weyl anomaly}\label{sec_weyl}
  In the previous section we constructed the partition function in the case where the background metric was the flat metric. Here we look at a metric $g =e^{\varphi} dx^2$ and our goal is to find a link between $ \Pi_{\gamma, \mu, \mu_{\partial}} ^{ (z_i, \alpha_i)_i, (s_j, \beta_j)_j}  (g, F) $ and $\Pi_{\gamma, \mu, \mu_{\partial}} ^{ (z_i, \alpha_i)_i, (s_j, \beta_j)_j}  ( dx^2, F)$. We claim that:

\begin{theorem}\label{WEYL} (Weyl anomaly)
 Given a metric $ g = e^{ \varphi} dx^2$ we have,\footnote{Our Weyl anomaly coincides with the one of \cite{gaw} but is slightly different from the one of \cite{Sphere, Disk}. This is because we have chosen to include the metric dependence on $g$ in the regularization procedure of section \ref{sec_circle}. }
 \begin{small}
 $$ \ln \frac{\Pi_{\gamma, \mu, \mu_{\partial}} ^{ (z_i, \alpha_i)_i, (s_j, \beta_j)_j}  ( g, F)}{\Pi_{\gamma, \mu,  \mu_{\partial}} ^{ (z_i, \alpha_i)_i, (s_j, \beta_j)_j}  ( dx^2,  F( \cdot - \frac{Q}{2} \varphi))} = -\sum_{i=1}^n \Delta_{\alpha_i} \varphi(z_i)  - \sum_{j=1}^{n'}  \frac{1}{2}\Delta_{\beta_j} \varphi(s_j) + \frac{1 + 6 Q^2}{96 \pi} \left( \int_{\Omega} | \partial \varphi|^2 d\lambda + 4 \int_{\partial \Omega} K \varphi  d \lambda_{\partial}  \right),$$
 \end{small}
where we have set $\Delta_{\alpha_i} = \frac{\alpha_i}{2}( Q -\frac{\alpha_i}{2})$ and $\Delta_{\beta_j} = \frac{\beta_j}{2}( Q -\frac{\beta_j}{2})$. 
 \end{theorem}
\begin{proof}
We need to factor out all the terms containing the metric $g$ in \eqref{partition}. We start by applying the Girsanov transform to the exponential term:
\begin{align*}
&\exp \left(- \frac{Q}{4 \pi} \int_{\Omega} R_g X_{g , \epsilon} d \lambda _g - \frac{Q}{2 \pi} \int_{\partial \Omega} K_g X_{g , \epsilon} d \lambda _{\partial g} \right)  \\
 &=\frac{Q}{4 \pi} \int_{\Omega} \Delta \varphi X_{g , \epsilon} d \lambda  - \frac{Q}{2 \pi} \int_{\partial \Omega} (K + \frac{1}{2}\partial_n \varphi) X_{g , \epsilon} d \lambda _{\partial }.
\end{align*}
Here we have used \eqref{change_R} and \eqref{change_K}. Using \eqref{gr_g} this shifts the field $X(x)$ by:
  \begin{align*}
  &\frac{Q}{4 \pi} \int_{\Omega} \Delta \varphi(y) G(x,y) d \lambda (y) - \frac{Q}{2 \pi} \int_{\partial \Omega} (K(y) + \frac{1}{2} \partial_n \varphi(y)) G(x,y) d \lambda _{\partial }(y) \\
 & = - \frac{Q}{2 } ( \varphi -  m_{\partial} ( \varphi)   ) 
  -  \frac{Q}{2 \pi} \int_{\partial \Omega} K(y) G(x,y) d \lambda _{\partial }(y) \\
 & = - \frac{Q}{2 } ( \varphi (x) -  m_{\partial } ( \varphi)   ) + Q \frac{ \ln (1/|x|) + \tau \ln (\tau/|x|)}{\tau + 1}. 
  \end{align*}
The term $m_{\partial } ( \varphi)$ does not depend on $x$ and can simply be removed by making a shift on $c$. The last term  $Q \frac{ \ln (1/|x|) + \tau \ln (\tau/|x|)}{\tau + 1}$ does not depend on the metric $g$ so it will simply cancel out when we take the ratio of the two partition functions. Therefore the only important shift we get is $-\frac{Q}{2} \varphi$. This changes  $F$ into $F(\cdot - \frac{Q}{2} \varphi)$ and using the formulas on Gaussian multiplicative chaos \eqref{shift_GMC} we see that $M_{\gamma, g, \epsilon}$ becomes $M_{\gamma, dx^2, \epsilon}$. Now we look at global factors that appear in front of our partition function. The shift on the insertions gives a term:
$$ \prod_{i=1}^n e^{- \frac{\alpha_i Q}{2} \varphi(z_i) } \prod_{j=1}^{n'} e^{- \frac{\beta_j Q}{4} \varphi(s_j) }. $$
The metric dependence of our regularization procedure applied to the insertions gives a term:
$$ \prod_{i=1}^n e^{ \frac{\alpha_i^2}{4} \varphi(z_i) } \prod_{j=1}^{n'} e^{ \frac{\beta_j^2}{8} \varphi(s_j) }. $$
We combine the two terms obtained above and get:
$$  \prod_{i=1}^n e^{ -\Delta_{\alpha_i} \varphi(z_i) } \prod_{j=1}^{n'} e^{ -\frac{1}{2}\Delta_{\beta_j} \varphi(s_j) }. $$
We now compute the global factor coming from the Girsanov transform. The variance of the curvature term is given by:
\begin{small}
\begin{align*}
&\frac{Q^2}{16 \pi^2} \int_{\Omega}\int_{\Omega} \Delta \varphi(x) \Delta \varphi(y) G(x,y) d\lambda(x) d\lambda(y) 
  + \frac{Q^2}{16 \pi^2} \int_{\partial\Omega}\int_{\partial\Omega} \partial_n \varphi(x) \partial_n \varphi(y) G(x,y) d\lambda_{\partial}(x) d\lambda_{\partial}(y) \\
  &- \frac{Q^2}{8 \pi^2}  \int_{\Omega}\int_{\partial\Omega} \Delta \varphi(x) \partial_n \varphi(y) G(x,y) d\lambda(x) d\lambda_{\partial}(y) + \frac{Q^2}{4 \pi^2} \int_{\partial\Omega}\int_{\partial\Omega} K(x) (K(y) + \partial_n \varphi(y))  G(x,y) d\lambda_{\partial}(x) d\lambda_{\partial}(y) \\
  &- \frac{Q^2}{4 \pi^2} \int_{\Omega}\int_{\partial\Omega} \Delta \varphi(x) K(y) G(x,y) d\lambda(x) d\lambda_{\partial}(y) \\
  &=  Q^2 \ln \tau + \frac{Q^2}{8 \pi} \int_{\Omega} |\partial \varphi|^2 d\lambda  + \frac{Q^2}{2 \pi} \int_{ \partial \Omega } \partial_n \varphi(x)  \frac{\ln |x| + \tau \ln(|x|/\tau)}{\tau +1}  d \lambda_{\partial}(x) - \frac{Q^2}{2 \pi} \int_{\Omega} \Delta \varphi(x) \frac{ \ln|x| + \tau \ln(|x|/\tau)}{\tau +1} d\lambda(x) 
\end{align*}
\end{small}
We apply the Green-Riemann formula \eqref{GR} on the last term, noticing that $\Delta \ln \vert x\vert =0$ and that $\partial_n ( \frac{\ln\vert x \vert + \tau \ln \vert x/ \tau \vert }{\tau +1} ) = K$:
$$- \frac{Q^2}{2 \pi} \int_{\Omega} \Delta \varphi(x) \frac{ \ln|x| + \tau \ln(|x|/\tau)}{\tau +1} d\lambda(x) = - \frac{Q^2}{2 \pi} \int_{ \partial \Omega } \partial_n \varphi(x)  \frac{\ln |x| + \tau \ln(|x|/\tau)}{\tau +1}  d \lambda_{\partial}(x) + \frac{Q^2}{2 \pi} \int_{ \partial \Omega } \varphi K d \lambda. $$
  Due to this Girsanov transform, we will have the exponential of the following expression in front of our partition function:
  $$ \frac{Q^2}{2} \ln \tau +  \frac{Q^2}{4 \pi} \int_{\partial \Omega} \varphi(x) K(x) d\lambda(x)  + \frac{Q^2}{16 \pi} \int_{\Omega} | \partial \varphi(x)|^2 d\lambda(x). $$
We must also add the term coming from $\mathcal{Z}_{GFF}(g)$ \eqref{shift_partition_GFF}:
  $$\exp\left(\frac{1}{96 \pi} (  \int_{\Omega} |\partial \varphi|^2 d\lambda + 4 \int_{\partial \Omega}  K \varphi d \lambda_{\partial} ) \right).$$
Putting everything together and taking the ratio of the partition functions we arrive at the Weyl anomaly.
 \end{proof}
The above result can be easily generalized to the case of two conformally equivalent metrics: 
 \begin{corollary}
 Given two metrics $g,g'$ linked by $g' = e^{\varphi} g$, we have:
 \begin{align*}
 &\ln \frac{\Pi_{\gamma, \mu, \mu_{\partial}} ^{ (z_i, \alpha_i)_i, (s_j, \beta_j)_j}  ( g', F)}{\Pi_{\gamma, \mu,  \mu_{\partial}} ^{ (z_i, \alpha_i)_i, (s_j, \beta_j)_j}  ( g,  F( \cdot - \frac{Q}{2} \varphi))}  \\
 &= -\sum_{i=1}^n \Delta_{\alpha_i} \varphi(z_i)  - \sum_{j=1}^{n'}  \frac{1}{2}\Delta_{\beta_j} \varphi(s_j) + \frac{1 + 6 Q^2}{96 \pi} \left( \int_{\Omega} | \partial^g \varphi|^2 d\lambda_g + 4 \int_{\partial \Omega} K_g \varphi  d \lambda_{\partial g} + 2 \int_{\Omega} R_g \varphi d\lambda_g \right).
 \end{align*}
 \end{corollary}
 \begin{proof}
 In this case we compute:
 \begin{small}
 \begin{align*}
 &\frac{1 + 6 Q^2}{96 \pi} \left(  4 \int_{\partial \Omega} K(x)  \varphi (x) d \lambda_{\partial}(x)  +  \int_{\Omega} | \partial \ln g'(x)|^2 d\lambda(x) - \int_{\Omega} | \partial \ln g(x)|^2 d\lambda(x) \right) \\
  &=  \frac{1 + 6 Q^2}{96 \pi} \left( 4 \int_{\partial \Omega} K(x)  \varphi (x) d \lambda_{\partial}(x)  +  \int_{\Omega} | \partial \varphi(x)|^2 d\lambda(x) + 2 \int_{\Omega} \partial \ln g(x) \cdot \partial\varphi(x) d\lambda(x) \right) \\
  &=  \frac{1 + 6 Q^2}{96 \pi} \left( 4 \int_{\partial \Omega} K(x)  \varphi (x) d \lambda_{\partial}(x)  +  \int_{\Omega} | \partial \varphi(x)|^2 d\lambda(x) - 2 \int_{\Omega} \Delta \ln g(x)  \varphi(x) d\lambda(x) + 2 \int_{\partial \Omega} \partial_n \ln g(x) \varphi(x) d\lambda(x) \right) \\
  &=  \frac{1 + 6 Q^2}{96 \pi} \left( \int_{\Omega} | \partial^g \varphi(x)|^2 d\lambda_g(x) + 4 \int_{\partial \Omega} K_g(x)  \varphi (x) d \lambda_{\partial g}(x)  + 2 \int_{\Omega} R_g(x)  \varphi(x) d\lambda_g(x) \right).
 \end{align*}
 \end{small}
 \end{proof}

  \subsection{KPZ formula}\label{sec_KPZ}
The KPZ formula was first introduced in \cite{KPZ} and was studied in a more probabilistic setting in \cite{DuSh}. In our case it consists in finding how the partition function behaves under the action of a conformal automorphism $\psi$ of our annulus meaning that we would like to give a relationship between $ \Pi_{\gamma, \mu, \mu_{\partial}} ^{ (\psi(z_i), \alpha_i)_i, (\psi(s_j), \beta_j)_j}  (dx^2, F) $ and $ \Pi_{\gamma, \mu, \mu_{\partial}} ^{ (z_i, \alpha_i)_i, (s_j, \beta_j)_j}  (dx^2, F) $. The proof we give is much simpler than the one in \cite{Sphere} or \cite{Disk} as we will just apply Theorem \ref{WEYL}. For this we start by  writing the invariance of our partition function under the change of coordinate $z \mapsto \psi(z)$:
  $$  \Pi_{\gamma, \mu, \mu_{\partial}} ^{ (\psi(z_i), \alpha_i)_i, (\psi(s_j), \beta_j)_j}  ( \vert \psi' \vert^2 dx^2, F( \cdot \circ \psi)) = \Pi_{\gamma, \mu, \mu_{\partial}} ^{ (z_i, \alpha_i)_i, (s_j, \beta_j)_j}  (dx^2, F). $$
This means that we can apply the Weyl anomaly formula with $\varphi = 2 \ln \vert \psi' \vert$. We compute:
$$ \int_{\Omega} | \partial \varphi|^2 d\lambda + 4 \int_{\partial \Omega} K \varphi  d \lambda_{\partial} = -\int_{\Omega} \varphi \Delta \varphi d\lambda + \int_{\partial \Omega} \partial_n \varphi  d \lambda_{\partial}  + 4
 \int_{\partial \Omega} K \varphi  d \lambda_{\partial}. $$
We recall that $\psi$ is either a rotation, the inversion $ z \mapsto \frac{\tau}{z}$, or the composition of both. Since the case of the rotation is trivial will assume that $\psi(z) = \frac{\tau}{z}$. We then have $\psi'(z) = - \frac{\tau}{z^2}$ thus $ \varphi(z)= 2 \ln \vert \frac{\tau}{z^2} \vert $. From this we get $\Delta \varphi =0$ and $ 4K + \partial_n \varphi =0$. Therefore the Weyl anomaly gives:
\begin{align*}
  &\Pi_{\gamma, \mu, \mu_{\partial}} ^{ (\psi(z_i), \alpha_i)_i, (\psi(s_j), \beta_j)_j}  ( \vert \psi' \vert^2 dx^2, F( \cdot \circ \psi)) = \nonumber \\
&\prod_{i=1}^n \vert \psi'(z_i) \vert^{2 \Delta_{\alpha_i}} \prod_{j=1}^{n'} \vert \psi'(s_j) \vert^{ \Delta_{\beta_j}} \Pi_{\gamma, \mu, \mu_{\partial}} ^{ (\psi(z_i), \alpha_i)_i, (\psi(s_j), \beta_j)_j}  (dx^2, F( \cdot \circ \psi + Q \ln \vert \psi' \vert )). 
\end{align*}
We arrive at:
\begin{theorem}\label{KPZ} (KPZ formula)
For any conformal automorphism $\psi$ of the annulus $\Omega$, the following holds:  
  $$ \Pi_{\gamma, \mu, \mu_{\partial}} ^{ (\psi(z_i), \alpha_i)_i, (\psi(s_j), \beta_j)_j}  (dx^2, F(\cdot \circ \psi + Q \ln \vert \psi' \vert)) = \prod_{i=1}^n |\psi'(z_i)|^{-2 \Delta_{\alpha_i}} \prod_{j=1}^{n'} |\psi'(s_j)|^{ -\Delta_{\beta_j}} \Pi_{\gamma, \mu, \mu_{\partial}} ^{ (z_i, \alpha_i)_i, (s_j, \beta_j)_j}  (dx^2, F). $$
  \end{theorem}
We can also give a similar formula for a conformal change of domain. Let $D$ be a (strict) domain of $\mathbb{C}$ with a smooth boundary and conformally equivalent to our annulus and let $ \psi : D \mapsto \Omega$ be a conformal map. We choose insertion points $(z_i, \alpha_i)$ in $D$ and $(s_j, \beta_j)$ in $\partial D$. We can define in the same way as \eqref{partition} the partition function of LQFT on the domain $D$:
\begin{align}\label{partition_D} 
   &\Pi_{\gamma,\mu,  \mu_{\partial}} ^{ (z_i, \alpha_i)_i, (s_j, \beta_j)_j}  (D, \tilde{g}, F) = \lim_{ \epsilon \rightarrow 0}
  \mathcal{Z}_{GFF}(\tilde{g}) \int_{\mathbb{R}} dc \: \mathbb{E} [F( X_{\tilde{g} , \epsilon} + c) \prod_{i=1}^n  \epsilon^{\frac{\alpha_i^2}{2}}  e ^{ \alpha_i ( X_{\tilde{g} , \epsilon} + c )(z_i)}  \prod_{j=1}^{n'}  \epsilon ^{\frac{\beta_j^2}{4}}  e ^{ \frac{\beta_j}{2} ( X_{\tilde{g} , \epsilon} + c )(s_j)}\nonumber \\
  &\times 
  \exp (- \frac{Q}{4 \pi} \int_{D} R_{\tilde{g}}(c + X_{\tilde{g} , \epsilon}) d \lambda _{\tilde{g}} 
  - \mu e^{\gamma c} M_{\gamma, \tilde{g}, \epsilon}(D) - \frac{Q}{2 \pi} \int_{\partial D} K_{\tilde{g}}(c + X_{\tilde{g} , \epsilon}) d \lambda _{\partial \tilde{g}} 
  - \mu_{\partial} e^{\frac{\gamma}{2} c}  M^{\partial}_{\gamma, \tilde{g}, \epsilon}( \partial D) )] 
  \end{align}
Here $\tilde{g}$ denotes a metric on $D$ and $X_{\tilde{g}}$ is the GFF on $D$ with Neumann boundary conditions and with vanishing mean over $\partial D$ in the metric $\tilde{g}$. We state the following result:
\begin{proposition} (Conformal change of domain)
Let $D$ be a domain of $\mathbb{C}$ with a smooth boundary and conformally equivalent to our annulus $\Omega$. Let $\psi : D \mapsto \Omega$ be a conformal map between $D$ and $\Omega$ and let $g_{\psi} = \vert \psi' \vert^2 g(\psi) $ be the pull-back of a metric $g$ on $\Omega$ by $\psi$. Then we have:
$$   \Pi_{\gamma, \mu, \mu_{\partial}} ^{ (z_i, \alpha_i)_i, (s_j, \beta_j)_j}  (D, g_{\psi}, F)  = \prod_{i=1}^n |\psi'(z_i)|^{2 \Delta_{\alpha_i}} \prod_{j=1}^{n'} |\psi'(s_j)|^{ \Delta_{\beta_j}} \Pi_{\gamma, \mu, \mu_{\partial}} ^{ (\psi(z_i), \alpha_i)_i, (\psi(s_j), \beta_j)_j}  (\Omega, g, F(\cdot \circ \psi + Q \ln \vert \psi' \vert)). $$ 
\end{proposition}

\subsection{Liouville field and measure at fixed $\tau$}\label{liouville_field}
Using our partition function we can now give the definition of the Liouville field $\phi_{\tau}$, the bulk Liouville measure $Z_{\tau}$, and the boundary Liouville measure $Z_{\tau}^{\partial}$.\footnote{Here we write the subscript $\tau$ to emphasize the fact that we are working at fixed $\tau$, i.e. on the annulus $\Omega$ of radii $1$ and $\tau$. In section \ref{LQG} we will integrate over $\tau \in (1, \infty)$ and define the general Liouville measures $Z$ and $Z^{\partial}$.}  These objects can be defined once the Seiberg bounds \eqref{sei1}-\eqref{sei3} are satisfied. Formally, $\phi_{\tau}$ is the log-conformal factor of the formal random metric $e^{\gamma \phi_{\tau}} g $ conformally equivalent to $g$.\footnote{The formal random metric $e^{\gamma \phi_{\tau}} g $ has been constructed rigorously by Miller-Sheffield for the special value $\gamma = \sqrt{\frac{8}{3}}$ corresponding to uniform planar maps, see for instance  \cite{Miller}.} The Liouville measure is a random measure that can be seen as the volume form of this formal metric tensor whereas the Liouville boundary measure corresponds to the line element along the boundary. More rigorously, given a measured space $E$, we denote $R(E)$ the space of Radon measures on $E$ equipped with the topology of weak convergence. For any metric $g$, the joint law of $(\phi_{\tau}, Z_{\tau}, Z_{\tau}^{\partial})$ is defined for continuous bounded functionals $F$ on $H^{-1}(\Omega)\times R(\Omega) \times R(\partial \Omega )$ by:
\begin{align*}
 &\mathbb{E}_{\gamma, \mu, \mu_{\partial},g} ^{ (z_i, \alpha_i)_i, (s_j, \beta_j)_j} [F(\phi_{\tau},Z_{\tau},Z_{\tau}^{\partial})] =
    (\Pi_{\gamma, \mu,  \mu_{\partial}} ^{ (z_i, \alpha_i)_i, (s_j, \beta_j)_j}  (g, 1)) ^{-1} \mathcal{Z}_{GFF}(g) \\
    &\times \lim_{\epsilon \rightarrow 0}  \int_{\mathbb{R}} \mathbb{E} [F( X_{g , \epsilon} + c, \epsilon^{\frac{\gamma^2}{2}} e^{\gamma(X_{g, \epsilon} + c )} d \lambda_g,  \epsilon^{\frac{\gamma^2}{4}} e^{\frac{\gamma}{2}(X_{g, \epsilon} + c )} d \lambda_{\partial g} )  \prod_i  \epsilon^{\frac{\alpha_i^2}{2}}  e ^{ \alpha_i ( X_{g , \epsilon} + c )(z_i)}  \prod_j  \epsilon^{\frac{\beta_i^2}{4}}  e ^{ \frac{\beta_j}{2} ( X_{g , \epsilon} + c)(s_j)}  \\
&\times 
  \exp (- \frac{Q}{4 \pi} \int_{\Omega} R_g(c + X_{g , \epsilon}) d \lambda _g 
  - \mu e^{\gamma c} M_{\gamma,g,\epsilon}(\Omega) - \frac{Q}{2 \pi} \int_{\partial \Omega} K_g(c + X_{g , \epsilon}) d \lambda _{\partial g} 
  - \mu_{\partial} e^{\frac{\gamma}{2} c} M^{\partial}_{\gamma, g,\epsilon}(\partial\Omega))]dc.    
\end{align*}

We can write a more compact expression for the joint law of the bulk and boundary measures in the case where $g =dx^2$. We introduce the notations,
\begin{align*}
& Z_0(dx) = e^{\gamma H(x)} M_{\gamma,dx^2}(dx),  \\
& Z_0^{\partial}(dx) = e^{\frac{\gamma}{2} H(x)} M^{\partial}_{\gamma, dx^2}(dx),
\end{align*}
and the ratio $ R = \frac{Z_0(\Omega)}{Z_0^{\partial}(\partial \Omega)^2}$. Then we get for the law of the bulk and boundary measures in the flat metric:
\begin{align*}
  &\mathbb{E}_{\gamma, \mu, \mu_{\partial},dx^2} ^{ (z_i, \alpha_i)_i, (s_j, \beta_j)_j} [F(Z_{\tau},Z_{\tau}^{\partial})]   \\
 & = \frac{1}{\mathcal{Z}} \int_{\mathbb{R}} e^{(\sum_i \alpha_i + \frac{1}{2} \sum_j \beta_j )c} \mathbb{E} [F(e^{\gamma c} Z_0, e^{ \frac{\gamma}{2}c}Z_0^{\partial})  \exp ( - \mu e^{\gamma c} Z_0(\Omega) - \mu_{\partial} e^{ \frac{\gamma}{2}c} Z_0^{\partial}(\partial \Omega))] dc \\
 &=  \frac{1}{\mathcal{Z}} \int_0^{\infty} y^{ \frac{2}{\gamma} (\sum_i \alpha_i + \frac{1}{2} \sum_j \beta_j ) -1} 
   \mathbb{E} [F(y^2 R \frac{Z_0}{Z_0(\Omega)}, y \frac{Z_0^{\partial}}{Z_0^{\partial}(\partial \Omega)}) \exp ( - \mu y^2 R - \mu_{\partial} y) Z_0^{\partial}(\partial \Omega)^{ -\frac{2}{\gamma} (\sum_i \alpha_i + \frac{1}{2} \sum_j \beta_j ) }  ] dy.
\end{align*}
The normalization constant $ \mathcal{Z} $ is computed by choosing $ F=1 $.
We can get an even simpler expression if we choose $\mu_{\partial} =0$, (a similar expression holds for $\mu=0$):
\begin{corollary}\label{mu_zero} Assume $\mu_{\partial} =0$. The joint law of the bulk/boundary measures are given by
\begin{align*}
&\mathbb{E}_{\gamma,  \mu, \mu_{\partial}=0, dx^2}^{(z_i,\alpha_i)_i,(s_j,\beta_j)_j}[F(Z_{\tau},Z_{\tau}^{\partial})] = \\
 &\frac{1}{\mathcal{Z}}  \int_0^{\infty} u^{ \frac{1}{\gamma} (\sum_i \alpha_i + \frac{1}{2} \sum_j \beta_j ) -1}  \mathbb{E} [F(u \frac{Z_0}{Z_0(\Omega)}, u^{1/2} \frac{Z_0^{\partial}}{Z_0( \Omega)^{1/2}})  Z_0( \Omega)^{ -\frac{1}{\gamma} (\sum_i \alpha_i + \frac{1}{2} \sum_j \beta_j ) }  ] e^{- \mu u} du
\end{align*}
where $\mathcal{Z}$ is the correct normalization constant in order to get a probability measure. In particular, the law of the volume of $\Omega$ follows a Gamma law with parameters $(\frac{ \sum_i \alpha_i +  \frac{1}{2}\sum_j \beta_j }{\gamma}, \mu)$.
\end{corollary}

\section{Liouville quantum gravity}\label{LQG}
\subsection{Convergence of the partition function}
We are now ready to construct the full theory of Liouville quantum gravity (LQG). Our goal is to integrate over $\tau$ the partition function of LQFT constructed in section \ref{LQFT} with an appropriate measure. As explained in section \ref{heuristics} of the appendix, the correct partition function of LQG is given by,
\begin{equation}
\mathcal{Z}_{LQG} = \int_1^{\infty} d \tau \mathcal{Z}_{Ghost}( \tau, dx^2) \mathcal{Z}_{Matter}(\tau, dx^2) \mathcal{Z}_{LQFT}(\tau, dx^2),
\end{equation}
where we have following expressions for the different partition functions,\footnote{In this section we modify very slightly the notation for the partition function of LQFT. Here we always choose $F=1$ and remove it from the notation but we include $\tau$ to show explicitly the dependence on $\tau$.}
\begin{align*}
&\mathcal{Z}_{GFF}( \tau, dx^2) =  \tau^{1/12} \frac{1}{ \vert \eta(\tau) \vert} \\
&\mathcal{Z}_{Matter}(\tau, dx^2) = (\mathcal{Z}_{GFF})^{c_m} =  \tau^{c_m/12} \frac{1}{\vert \eta (\tau) \vert^{c_m} } \\
&\mathcal{Z}_{Ghost}(\tau, dx^2)  = \frac{1}{\tau} \tau^{-\frac{13}{6}} \vert \eta (\tau) \vert ^2 \\
&\mathcal{Z}_{LQFT}(\tau, dx^2) = \int_{\partial \Omega} d \lambda_{\partial}(s) \Pi^{ ( s, \gamma ) }_{\gamma, \mu, \mu_{\partial}} (\tau, dx^2),
\end{align*} 
with
\begin{equation}
\eta (\tau) = \tau^{-1/12} \prod_{n=1}^{\infty} (1 - \tau^{-2n} ) .
\end{equation}
The real constant $c_m$ is the central charge of the matter fields. It is link to $Q$ by the relation written in the appendix \eqref{scaling}:
\begin{equation}\label{central_charge}
c_m - 25 + 6 Q^2 =0.
\end{equation}
Let us explain the expression for $\mathcal{Z}_{LQFT}$. It corresponds to the partition function of LQFT that we have constructed in section \ref{LQFT} with one insertion point on the boundary (this is the minimal condition required for the Seiberg bounds (\ref{sei1})-(\ref{sei3}) to be satisfied). We choose an insertion point on the boundary $s \in \partial \Omega$ with weight $\gamma$ and we integrate the position of $s$ over the whole boundary. This is the correct choice to get a link with random planar maps, see section \ref{sec_maps}. From all of this we get:
\begin{align}
\mathcal{Z}_{LQG} &= \int_{\partial \Omega}  d \lambda_{\partial}(s)  \int_1^{\infty} d \tau \mathcal{Z}_{Ghost}( \tau, dx^2) \mathcal{Z}_{Matter}(\tau, dx^2) \Pi_{\gamma, \mu, \mu_{\partial}}^{ ( s, \gamma ) } (\tau, dx^2) \\
&= \int_{\partial \Omega}  d \lambda_{\partial}(s)  \int_1^{\infty} \frac{d \tau}{\tau} \tau^{ \frac{c_m -25}{12}} \vert \eta ( \tau ) \vert ^{1 -c_m} \mathcal{Z}_{GFF}^{-1} \Pi_{\gamma, \mu, \mu_{\partial}}^{ ( s, \gamma ) } (\tau, dx^2). \nonumber
\end{align}
The integral over the position of the boundary insertion can be easily simplified using the KPZ formula of section \ref{sec_KPZ}. For any $s \in \partial \Omega $, we can write $ s = \psi(1)$ for some conformal automorphism $\psi$. A computation gives:
\begin{equation}
\int_{\partial \Omega}  d \lambda_{\partial}(s) \Pi_{\gamma, \mu, \mu_{\partial}}^{(s,\gamma)}( \tau, dx^2)  = 2 \pi (1 + \tau^{1 - \frac{\gamma}{2} ( Q - \frac{\gamma}{2}) }) \Pi_{\gamma, \mu, \mu_{\partial}}^{(1,\gamma)} ( \tau, dx^2) = 4 \pi \Pi_{\gamma, \mu, \mu_{\partial}}^{(1,\gamma)} ( \tau, dx^2).
\end{equation}
In the following we will drop the irrelevant factor $4 \pi$. Our goal is now to prove:
\begin{theorem}
Let us assume that $\mu_{\partial} > 0$. Then the partition function of Liouville quantum gravity given by
$$ \mathcal{Z}_{LQG} = \int_1^{\infty} \frac{d \tau}{\tau} \tau^{ \frac{c_m -25}{12}} \vert \eta ( \tau ) \vert ^{1 -c_m} \mathcal{Z}_{GFF}^{-1} \Pi_{\gamma, \mu, \mu_{\partial}}^{ ( 1, \gamma ) } (\tau, dx^2)  $$
is finite for any value of $\gamma \in (0,2)$.
\end{theorem}
\begin{proof}
Let us first note that thanks to relation \eqref{central_charge} $\gamma < 2$ is equivalent to $c_m <1$. We will look separately at the behaviour of our integral for $ \tau \rightarrow \infty$ and for $ \tau \rightarrow 1$. We write $\mathcal{Z}_{LQG} = \mathcal{Z}_{LQG}^1 + \mathcal{Z}_{LQG}^{\infty} $ with:
\begin{align*}
&\mathcal{Z}_{LQG}^1  = \int_1^2 \frac{d \tau}{\tau} \tau^{ \frac{c_m -25}{12}} \vert \eta ( \tau ) \vert ^{1 -c_m} \mathcal{Z}_{GFF}^{-1} \Pi_{\gamma, \mu, \mu_{\partial}}^{ ( 1, \gamma ) } (\tau, dx^2), \\
&\mathcal{Z}_{LQG}^\infty  = \int_2^{\infty} \frac{d \tau}{\tau} \tau^{ \frac{c_m -25}{12}} \vert \eta ( \tau ) \vert ^{1 -c_m} \mathcal{Z}_{GFF}^{-1} \Pi_{\gamma, \mu, \mu_{\partial}}^{ ( 1, \gamma ) } (\tau, dx^2).
\end{align*}
Throughout this proof we will use the abuse of notation $ e^{ \gamma X(x) } d \lambda(x) $ and $ e^{ \frac{\gamma}{2} X(x)} d \lambda_{\partial}(x) $ for $M_{\gamma, dx^2}(dx)$ and $M^{\partial}_{\gamma, dx^2}(dx)$. Using \eqref{partition_2} we can write,
$$ \frac{\Pi_{\gamma, \mu, \mu_{\partial}}^{ ( 1, \gamma ) } (\tau, dx^2)}{\mathcal{Z}_{GFF}(\tau, dx^2)} = e^{C(1)} \int_{\mathbb{R}} e^{\frac{\gamma c}{2} } \mathbb{E} [ \exp ( - \mu e^{\gamma c} \int_{\Omega} e^{\gamma H + \gamma X} d \lambda - \mu_{\partial} e^{\frac{\gamma c}{2} } \int_{ \partial \Omega } e^{ \gamma H /2 + \gamma X/2 } d \lambda_{\partial} )] dc, $$ 
with:
\begin{align*}
&H(x) = \frac{\gamma}{2} G(x,s) - Q \ln \vert x \vert, \\
&C(s) = \frac{\gamma^2}{8} h_{\partial}(s) + \frac{Q^2}{2} \ln \tau  - \frac{Q \gamma}{2} \ln \vert s \vert.
\end{align*}
Therefore
$$ C(1) = \frac{\gamma^2}{8} h_{\partial}(1) + \frac{Q^2}{2} \ln \tau  = \frac{Q^2}{2} \ln \tau+  \frac{\gamma^2}{8} \left( \frac{\tau^2 \ln \tau}{(1+\tau)^2} +  \sum_{n=1}^{\infty} \frac{\tau^{2n}+1}{n \tau^{2n} (\tau^{2n}-1)} + \ln \frac{\tau^4}{(\tau^2-1)^2}\right) + \tilde{C}  $$
where $\tilde{C}>0$ is some constant.
We restrict ourselves to the integral over the inner boundary $\partial \Omega_1$ and we use the following upper bound:
\begin{align*}
\frac{\Pi_{\gamma, \mu, \mu_{\partial}}^{(1, \gamma)}(\tau, dx^2)}{\mathcal{Z}_{GFF} (\tau, dx^2)} &\leq e^{C(1)} \int_{\mathbb{R}} e^{\frac{\gamma c}{2} } \mathbb{E} [ \exp (  - \mu_{\partial} e^{\frac{\gamma c}{2} } \int_{ \partial \Omega_1 } e^{ \gamma H /2 + \gamma X/2 } d \lambda_{\partial} )] dc \\
&\leq e^{C(1)} \frac{2}{\gamma} \frac{1}{\mu_{\partial}}  \mathbb{E} \left[ \frac{1}{ \int_{ \partial \Omega_1} e^{ \gamma H/2 + \gamma X/2} d \lambda_{\partial}} \right]. 
\end{align*}
We will decompose our process $X$ into two parts $Y$ and $Z$ according to the following decomposition of the Green's function, written for $z=e^{i \theta} $, $z' = e^{i \theta'}$:
\begin{align*}
&G_Y(z,z') = g_0(1,1) +2\sum_{n=1}^{\infty} {g}_n(1,1) \cos n(\theta - \theta') = \frac{\tau^2 \ln \tau}{ (1+ \tau)^2} + \sum_{n=1}^{+\infty} \frac{\tau^{2n}+1}{ n \tau^{2n}(\tau^{2n} -1) } \cos n (\theta - \theta'), \\
&G_Z(z,z') = \ln \frac{|\tau^4 z^2 z'^2|}{|1 - z \overline{z'}| |\tau^2 - z \overline{z'}||z - z'| |\tau^2 z -  z'|} = \ln \frac{\vert \tau^4 \vert }{\vert e^{i \theta} - e^{i \theta'} \vert^2 \vert \tau^2 e^{i \theta} - e^{i \theta'} \vert^2  }.
\end{align*}\\
\textbf{Step 1: $\tau \rightarrow \infty$} \\
We start by looking at the behaviour of our integral when $\tau \rightarrow \infty$. We see a diverging term $\frac{\tau^2 \ln \tau}{(1+\tau)^2}$ in our correlation function $G_Y$ but this is simply a constant which means it corresponds to a constant Gaussian variable independent of everything else. We can remove it simply by making a shift on $c$ in \eqref{partition_2}. This also removes the corresponding term in $C(1)$. For the remaining part of $G_Y$, we have uniformly on $\partial \Omega_1$:
$$\sum_{n=1}^{+\infty} \frac{\tau^{2n}+1}{ n \tau^{2n}(\tau^{2n} -1) } \cos n (\theta - \theta') \underset{\tau \rightarrow \infty} {\rightarrow} 0.  $$
For $G_Z$ we have in the same way, uniformly on $\partial \Omega_1$:
$$ \ln \frac{\vert \tau^4 \vert }{ \vert \tau^2 e^{i \theta} - e^{i \theta'} \vert^2  }\underset{\tau \rightarrow \infty} {\rightarrow}  0.  $$
This means that by using Kahane's inequalities \eqref{ineq} we can bound these terms that converge to $0$ and we are left only with $ \ln \frac{1}{ \vert z -z' \vert^2}$ in our covariance function. We call $\hat{Z}$ a Gaussian process of covariance $ \ln \frac{1}{ \vert z -z' \vert^2}$. Similarly for $C(1)$, as $\tau \rightarrow  \infty$ the only terms that remain are $\frac{Q^2}{2} \ln \tau$ and the constant $\tilde{C}>0$. Concerning $\eta(\tau)$ we bound it in this case by $\frac{1}{\tau^{12}}$.  Combining everything we get, for some $C_1>0$:
$$ \mathcal{Z}_{LQG}^{\infty}  \leq C_1 \int_2^{\infty} d \tau \tau^{ \frac{c_m -25 +6Q^2}{12}} \frac{1}{\tau^{1 + \frac{1-c_m}{12}}}   \mathbb{E} \left[ \frac{1}{ \int_{ \partial \Omega_1} e^{ \frac{\gamma}{2} \hat{Z}(x)} \vert x -1 \vert^{-\frac{\gamma^2}{2}}   d \lambda_{\partial}(x)} \right].   $$
The theory of Gaussian multiplicative chaos \cite{review} tells us that the above expectation $\mathbb{E}[\cdot]$ is finite and it is clearly independent of $\tau$. Using the relation \eqref{central_charge} we can simplify the powers of $\tau$: we are left with $\tau^{-1 - \frac{1-c_m}{12}} $  which is integrable for $c_m<1$. Therefore the integral converges for $\tau \rightarrow \infty$.\\

\noindent
\textbf{Step 2: $\tau \rightarrow 1$}\\
In this case we have to be more careful as $G_Y$ has a divergence in $\tau$ that depends on $z$. Indeed the sum $\sum_{n=1}^{+\infty} \frac{\tau^{2n}+1}{ n \tau^{2n}(\tau^{2n} -1) }$ diverges as $\frac{\pi^2}{6} \frac{1}{\tau -1}$ as $\tau \rightarrow 1$. Therefore instead of integrating over the full inner boundary $\partial \Omega_1$, we will restrict ourselves to a small open neighbourhood $V(\epsilon)$ of $1$ such that for $ x, y \in V(\epsilon)$, $\vert G_Y(x,y) - G_Y(1,1) \vert \leq  \epsilon \; G_Y(1,1) $ for a fixed $\epsilon>0$. For this bound to hold the size of $V(\epsilon) $ needs to scale as  $ \ln \tau $ with $ \tau $.
We then get:
\begin{small}
\begin{align*}
 \mathbb{E}&\left[ \left(\int_{\partial \Omega_1} e^{ \frac{\gamma}{2} (Y(x) +  Z(x)) + \frac{\gamma^2}{4}(G_Y(x,1) + G_Z(x,1) )  } d\lambda_{\partial}(x)\right)^{-1}\right] \\
 &\leq \mathbb{E}\left[ \left(\int_{V(\epsilon)} e^{ \frac{\gamma}{2} (Y(x) + Z(x)) + \frac{\gamma^2}{4}(G_Y(x,1) + G_Z(x,1) )  } d\lambda_{\partial}(x) \right)^{-1} \right] \\
 &\leq  e^{- \frac{\gamma^2}{4} (1 - \epsilon) G_Y(1,1) } \mathbb{E}\left[ \left( \int_{V(\epsilon)} e^{ \frac{\gamma}{2} (Y(x) + Z(x)) +  \frac{\gamma^2}{4}G_Z(x,1)   } d\lambda_{\partial}(x) \right)^{-1} \right] \\ 
 &\leq  e^{- \frac{\gamma^2}{4} (1 - \epsilon) G_Y(1,1) + \frac{\gamma^2}{8} (1 + \epsilon)  G_Y(1,1) } \mathbb{E}\left[ \left( \int_{V(\epsilon)} e^{  \frac{\gamma}{2} Z(x) +  \frac{\gamma^2}{4}G_Z(x,1)   } d\lambda_{\partial}(x) \right)^{-1} \right].
\end{align*}
\end{small}
In the last inequality we have used Kahane's inequalities \eqref{ineq} to replace the $e^{\frac{\gamma}{2} Y(x)}$ in the denominator by  $e^{\frac{\gamma}{2} \sqrt{1 + \epsilon} Y(1) }$. For the covariance $G_Z$, the log term $\ln \frac{1}{\vert \tau^2 e^{i \theta} - e^{i \theta'} \vert^2 }$ also diverges as $\tau \rightarrow 1$. We will simply bound it by $\ln \frac{1}{\vert \tau^2 -1 \vert^2 }$. Using \eqref{ineq} this gives a polynomial divergence which is irrelevant compared to the exponential terms. Finally there is a last exponential divergence coming from the constant $C(1)$ which is of the order of $ e^{\frac{\gamma^2}{8}   G_Y(1,1) } $.  Calling again $\hat{Z}$ a Gaussian process of covariance $\ln \frac{1}{\vert z - z' \vert^2}$, we get for some $C_2>0$:
$$\mathcal{Z}_{LQG}^{1} \leq C_2 \int_1^2 d \tau \vert \eta(\tau) \vert^{1 -c_m} e^{ \frac{ 3 \epsilon \gamma^2}{8} G_Y(1,1) }  \mathbb{E} \left[ \left( \int_{ V(\epsilon) } e^{ \frac{\gamma}{2} \hat{Z}(x)} \vert x -1 \vert^{-\frac{\gamma^2}{2}}   d \lambda_{\partial}(x) \right)^{-1} \right].    $$
To finish we use the following bound on $ \eta(\tau) $:
$$  \eta (\tau) = \tau^{-1/12} \prod_{n=1}^{\infty} (1 - \tau^{-2n} ) \leq  \tau^{-1/12} e^{ - \frac{1}{\tau^2 -1} }. $$
Therefore we get for some $C_3>0$:
$$\mathcal{Z}_{LQG}^{1} \leq C_3 \int_1^2 d \tau  e^{ (\frac{\pi^2 \gamma^2}{16} \epsilon -(\frac{1-c_m}{2})) \frac{1}{\tau-1}}   \mathbb{E} \left[ \left( \int_{ V(\epsilon) } e^{ \frac{\gamma}{2} \hat{Z}(x)} \vert x -1 \vert^{-\frac{\gamma^2}{2}}   d \lambda_{\partial}(x) \right)^{-1} \right].    $$
$V(\epsilon) $ scales as  $ \ln \tau $ with $ \tau $ so the above expectation $\mathbb{E}[\cdot]$ produces a logarithmic divergence which is again negligible compared to the exponential terms (see for instance \cite{Houches}). Since $c_m<1$, this quantity is bounded if we choose $\epsilon$ small enough. Therefore $\mathcal{Z}_{LQG} < + \infty$.
\end{proof}

\subsection{Joint law of the Liouville measures and of the random modulus}
Our goal in this section is to give the joint law of the bulk Liouville measure $Z$, the boundary Liouville measure $Z^{\partial}$, and of the random $\tau$\footnote{In the literature the modulus of an annulus of radii $1$ and $\tau$ is given by $\frac{1}{2 \pi} \ln \tau  $ but here we choose to write everything with our parameter $\tau$.}. Because the annulus $\Omega = \{ z , 1 < |z| < \tau \} $ depends explicitly on $\tau$, when we perform an integration on $\tau$ this poses a problem to have a fix domain on which to define $Z$ and $Z^{\partial}$. To solve this problem we will work on a fixed annulus $\hat{\Omega} = \{ z , 1 < |z| < 2 \} $ which is simply the annulus of radii $1$ and $2$. We then introduce the map $f_{\tau} : \hat{\Omega} \mapsto \Omega $ with expression given in polar coordinates $z = r e^{i \theta}$ by:
\begin{equation}\label{scaling_f}
f_{\tau}(r e^{i \theta}) = ( (\tau -1)(r-1) +1 ) e^{i \theta}.
\end{equation}
Of course $f_{\tau}$ will not be a conformal map. For any Borel set $A \subset \hat{\Omega}$, we will write $f_{\tau}(A)$ for its image under $f_{\tau}$ (and similarly for subsets of $\partial \hat{\Omega}$). Just like in the previous subsection we will work with one insertion point of weight $\gamma$ placed at the point $1$ on the inner boundary of $\hat{\Omega}$ (but this choice is irrelevant thanks to the KPZ relation).
%\begin{align*}
%& Z_0(dx^2) = e^{\gamma H(x)} e^{\gamma X(x)} dx^2, \quad Z_0^{\partial}(dx) = e^{\frac{\gamma}{2} H(x)}  e^{\frac{\gamma}{2} X(x)}dx \\
%& \mathbb{E}[X(x)X(y)] = G(x,y), \quad  H(x) = \frac{\gamma}{2} G(x,1) - Q \ln \vert x \vert, \quad R = \frac{Z_0(\Omega)}{Z_0^{\partial}( \partial \Omega )^2} \\
%& \eta (\tau) = \tau^{-1/12} \prod_{n=1}^{\infty} (1 - \tau^{-2n} ), \quad c_m = 25 - 6 Q^2
%\end{align*}
We are now ready to give the expression of the joint law of $(Z, Z^{\partial}, \tau)$. For any Borel sets $A \subset \hat{\Omega}$, $B \subset \partial \hat{\Omega}$ and for any continuous bounded functionals $F$ we have:
\begin{align*}
\mathbb{E}^{\gamma}_{\mu, \mu_{\partial}}[F(Z(A), Z^{\partial}(B), \tau) ] &= \frac{1}{\mathcal{Z}} \int_1^{\infty} d \tau \tau^{ \frac{c_m -25}{12} -1} \vert \eta ( \tau ) \vert ^{1 -c_m} \\
&\times \int_0^{\infty} dy \mathbb{E} [F(y^2 R \frac{Z_0(f_{\tau}(A))}{Z_0(\Omega)}, y \frac{Z_0^{\partial}(f_{\tau}(B))}{Z_0^{\partial}(\partial \Omega)}, \tau) \exp ( - \mu y^2 R - \mu_{\partial} y) Z_0^{\partial}(\partial \Omega)^{ -1 }  ]. 
\end{align*}

In this expression and in the following ones the normalization constant $\mathcal{Z}$ is always computed by choosing $F=1$. The definitions of $Z_0$, $Z_0^{\partial}$ and $R$ are given in section \ref{liouville_field}. By applying the above formula we can give the joint law of the total inner boundary, total outer boundary and of the random $\tau$:
\begin{align*}
\mathbb{E}^{\gamma}_{\mu, \mu_{\partial}}[F(Z^{\partial}(\partial \hat{\Omega}_1), Z^{\partial}(\partial \hat{\Omega}_{2}), \tau) ] &= \frac{1}{\mathcal{Z}} \int_1^{\infty} d \tau \tau^{ \frac{c_m -25}{12} -1} \vert \eta ( \tau ) \vert ^{1 -c_m} \\
& \times  \int_0^{\infty} dy \mathbb{E} [F(y \frac{Z_0^{\partial}(\partial \Omega_1)}{Z_0^{\partial}(\partial \Omega)}, y \frac{Z_0^{\partial}(\partial \Omega_{\tau})}{Z_0^{\partial}(\partial \Omega)}, \tau) \exp ( - \mu y^2 R - \mu_{\partial} y) Z_0^{\partial}(\partial \Omega)^{ -1 }  ].
\end{align*}
Simplifying the above we get the law of $\tau$:
$$ \mathbb{E}^{\gamma}_{\mu, \mu_{\partial}}[F( \tau) ] = \frac{1}{\mathcal{Z}} \int_1^{\infty} d \tau \tau^{ \frac{c_m -25}{12} -1} \vert \eta ( \tau ) \vert ^{1 -c_m}  \int_0^{\infty} dy \mathbb{E} [F(\tau) \exp ( - \mu y^2 R - \mu_{\partial} y) Z_0^{\partial}(\partial \Omega)^{ -1 }  ]. $$
We can also give a joint law conditioned on the total boundary length $Z^{\partial}(\partial \hat{\Omega})$, for some $L>0$:
\begin{align*}
 &\mathbb{E}^{\gamma}_{\mu, \mu_{\partial}}[F(Z(A), Z^{\partial}(B), \tau)\vert  Z^{\partial}(\partial \hat{\Omega}) = L] = \\
 &\frac{1}{\mathcal{Z}} \int_1^{\infty} d \tau \tau^{ \frac{c_m -25}{12} -1} \vert \eta ( \tau ) \vert ^{1 -c_m}   \mathbb{E} [F(L^2 R \frac{Z_0(f_{\tau}(A))}{Z_0(\Omega)}, L \frac{Z_0^{\partial}(f_{\tau}(B))}{Z_0^{\partial}(\partial \Omega)}, \tau) \exp ( - \mu L^2 R - \mu_{\partial} L) Z_0^{\partial}(\partial \Omega)^{ -1 }  ]. 
\end{align*}
If we were to condition on $\tau$ we would recover the formulas of section \ref{liouville_field}.

\subsection{Conjectured link with random planar maps}\label{sec_maps}

The theory of Liouville quantum gravity that we have constructed is the conjectured limit of random planar maps potentially weighted by a certain statistical physics model. Let us try to state a precise mathematical conjecture. We call $Q_{n,p}$ the set of quadrangulations with the topology of an annulus, $n$ inner faces, a perimeter of $2p$ and with one marked point on the boundary. The growth of $Q_{n,p}$ is expected to be of the order $ \mu_c^n  \mu_{\partial c}^{2p}$ times a polynomial term. The constants $\mu_c$ and $\mu_{\partial c} $ are non universal and they depend on the discrete model chosen (the values change for instance for triangulations). Now we fix an $a>0$ and introduce $ \overline{\mu} = \mu_c + a^2 \mu $, $\overline{\mu}_{\partial} = \mu_{ \partial c} + a \mu_{\partial}$. We can give a conformal structure to each $Q \in Q_{n,p} $ and conformally map it to an annulus of radii $1$ and $\tau$ for some unique $\tau$ where the marked point of $Q$ is mapped to $1$. Then by using the function $f_{\tau}$ \eqref{scaling_f} we can map $Q$ to the reference annulus $\hat{\Omega}$ defined in the previous subsection. By using this mapping, for each $Q \in Q_{n,p} $ we define the bulk and boundary measures $\nu_{Q,a}$ and $\nu^{\partial}_{Q,a}$ on $\hat{\Omega}$ by giving a volume $a^2$ to each face of $Q$ and a length $a$ to each edge on the boundary. We now define the random measures $ (\nu_a, \nu_a^{\partial}) $ by the relation, for all suitable functionals $F$:
$$ \mathbb{E}^a [ F(\nu_a, \nu_a^{\partial})] := \frac{1}{\mathcal{Z}_a} \sum_{N,p} e^{- \overline{\mu} N} e^{- \overline{\mu}_{\partial} 2p} \sum_{Q \in Q_{N,p}} F(\nu_{Q,a}, \nu_{ Q,a}^{\partial}).$$
The constant $\mathcal{Z}_a$ is again the normalization constant. We now state:
\begin{conjecture} The limit in law when $ a\rightarrow 0$ of $ (\nu_a, \nu_a^{\partial}) $ exists in the space of Radon measures equipped with the topology of weak convergence and is given by the measures $Z$ and $Z^{\partial}$ with $\gamma =\sqrt{\frac{8}{3}}$ and with appropriate cosmological constants $\mu$ and $\mu_{\partial}$. More precisely we have
$$ \lim_{a \rightarrow 0} \mathbb{E}^a [ F(\nu_a, \nu_a^{\partial})] = \mathbb{E}^{\sqrt{\frac{8}{3}}}_{ \mu, \mu_{\partial} }[ F( Z, Z^{\partial} )] $$
for a certain value of $\mu$ and $\mu_{\partial}$.
\end{conjecture}
If we couple the planar maps to a statistical physics model, a similar conjecture is expected to hold but for a different value of $\gamma$ (see \cite{Tori} for a discussion of this in the case of the torus).

\section{Appendix}

\subsection{Liouville quantum gravity and conformal field theory}\label{heuristics}
Let us give some heuristic ideas on the theory of 2D quantum gravity and on conformal field theory. For a two-dimensional surface $M$ let $\mathcal{M}$ be the set of (smooth) Riemannian metrics on $M$. The partition function of (euclidean) 2D quantum gravity on the surface $M$ is given by:
\begin{equation}
\mathcal{Z}_{LQG} = \int_{\mathcal{M}} Dg e^{- S_{EH}(g) - S_{Matter}(g) }.
\end{equation}
Here $S_{EH}$ is Einstein-Hilbert action with a cosmological constant $\mu_0 >0$ 
\begin{equation}
S_{EH}(g) = \int_M d\lambda_g (R_g + \mu_0)
\end{equation}
where $R_g$ stands for the scalar curvature of the metric $g$. $S_{Matter}$ is the action of the matter fields of the theory, we will discuss its expression below. $Dg$ is a formal uniform measure on the space $\mathcal{M}$ and a major difficulty is precisely to give sense to this measure. In two dimensions there are several simplifications. First of all, due to the Gauss-Bonnet theorem we have $\int_M d\lambda_g  R_g = 8\pi (1 -h)$ where $h$ is the genus of the surface $M$ meaning that this term can be factored out of the integration over $\mathcal{M}$. Next we have a particularly simple description of the space $\mathcal{M}$ in two-dimensions. It turns out that any $g \in \mathcal{M}$ can be written in the simple form $g = \psi(e^{\phi} \hat{g}_{\tau})$, where $\psi$ is diffeomorphism (a change of coordinates), $\phi: M \rightarrow \mathbb{R}$ is a function called the Weyl factor, and $\hat{g}_{\tau}$ are a set a reference metrics parametrized by $\tau$ which correspond to the moduli space of $M$ (the set of non-equivalent conformal structures on $M$). The description of this moduli space can be quite complicated for $M$ of arbitrary topology but it remains finite dimensional. With this result in mind it is natural to want to write the formal measure $Dg$ in the following way
\begin{equation}
Dg = D \psi D \phi D \tau J(\phi, \tau )
\end{equation}
where $J(\phi, \tau )$ is the infinite dimensional Jacobian of this change of variable. Considering that $S_{EH}$ and $S_{Matter}$ are independent of the choice of coordinates (like any relevant quantity in physics), and so is the Jacobian $J$, we can simply drop the integration $D \psi$ over the diffeomorphisms. $D \tau$ is the finite dimensional integration measure over the moduli space so it has a well defined expression. We now come to the difficult problem of giving the expression for $D\phi J( \phi, \tau)$. 
It is heuristically shown in the physics literature \cite{David}, \cite{DPhong}, \cite{DistKa} that this quantity can be expressed using the Liouville action. We introduce
\begin{equation}
\mathcal{Z}_{LQFT}(\hat{g}_{\tau}) =  \int_{\Sigma} D_{\hat{g}_{\tau}}X e^{- S_L(X, \hat{g}_{\tau})}
\end{equation}
where again $D_{\hat{g}_{\tau}} X$ is a formal uniform measure on the space $\Sigma$ of functions $X: M \rightarrow \mathbb{R}$. This partition function $\mathcal{Z}_{LQFT}$ is precisely (up to the insertion points) our partition function $\Pi_{\gamma,\mu,\mu_\partial}^{(z_i,\alpha_i), (s_j,\beta_j)}(\hat{g}_{\tau},1)$. In the same way it is possible to write the matter term as a partition function $\mathcal{Z}_{Matter}(\hat{g}_{\tau})$ and to write the formal Jacobian $J$ as what physicists call the ghost partition function $\mathcal{Z}_{Ghost}(\hat{g}_{\tau})$. The black box that we use is the following expression for $\mathcal{Z}_{LQG}$:  
\begin{hypothesis}
The formal ill-defined path integral $\int_{\mathcal{M}} Dg e^{- S_{EH}(g) - S_{Matter}(g) }$ can be expressed in the following way:
\begin{equation}
\mathcal{Z}_{LQG} = \int D\tau \mathcal{Z}_{Ghost}(\hat{g}_{\tau}) \mathcal{Z}_{Matter}(\hat{g}_{\tau}) \mathcal{Z}_{LQFT}(\hat{g}_{\tau}).
\end{equation}
The integral $\int D \tau$ is over the moduli space of the surface $M$.
Furthermore the parameters entering in the definitions of $\mathcal{Z}_{Matter}$ and $\mathcal{Z}_{LQFT}$ must be chosen in such a way that the value of $\mathcal{Z}_{LQG}$ is independent of the choice of reference metrics $\hat{g}_{\tau}$ (see below).
\end{hypothesis}
Making rigorous sense of this statement is an open problem, see for example \cite{DHoker}, \cite{Mavro}. With this new formulation we are now in the formalism of conformal field theory (CFT). Each of the partition functions $\mathcal{Z}_{Ghost}(g)$, $\mathcal{Z}_{Matter}(g)$ and $\mathcal{Z}_{LQFT}(g)$ (and also  $\mathcal{Z}_{GFF}(g)$ of \eqref{partition_GFF_def}) is expected to obey the following rule when we rescale the metric $g$ by some Weyl factor $e^{\sigma}$:
\begin{equation}
\mathcal{Z}_{CFT}(e^{\sigma} g)  = e^{ \frac{c}{96 \pi} \int_M \vert \partial^g \sigma \vert^2 d\lambda_g + 2 \int_M R_g \sigma d \lambda_g } \mathcal{Z}_{CFT}(g).
\end{equation}
This formula holds for a surface $M$ with no boundary, see \cite{gaw}. In the case of a surface with boundary $\Omega$ like our annulus, there is an extra boundary term:
\begin{equation}\label{weyl_boun}
 \mathcal{Z}_{CFT}(e^{\sigma} g)  = e^{ \frac{c}{96 \pi} \int_{\Omega} \vert \partial^g \sigma \vert^2 d\lambda_g + 2 \int_{\Omega} R_g \sigma d \lambda_g + 4 \int_{\partial \Omega} K_g \sigma d \lambda_{\partial g} } \mathcal{Z}_{CFT}(g). 
\end{equation}
The constant $c$ is the central charge of the conformal field theory. We have $c_{GFF} =1$, $c_{LQFT} = 1 + 6Q^2$ and $c_{Ghost} = -26$. The central charge of the matter $c_m$ can vary depending on what matter fields we choose. These formulas have an important consequence: since the whole partition function $\mathcal{Z}_{LQG}$ must be independent of our choice of the reference metrics $\hat{g}_{\tau}$, we ask $\mathcal{Z}_{LQG}$ to be invariant by Weyl rescaling (namely if we replace $\hat{g}_{\tau}$ by $ e^{\sigma} \hat{g}_{\tau}$ for some $e^{\sigma}$). This imposes the relation:
\begin{equation}\label{scaling}
c_m + c_{LQFT} + c_{Ghost} = 0 \Rightarrow c_m - 25 + 6 Q^2 = 0.
\end{equation}
This gives a relation between  $c_m$ and the parameter $Q = \frac{\gamma}{2} + \frac{2}{\gamma}$. Finally let us discuss the choice for $\mathcal{Z}_{Matter}$. The most common is to choose a power of the free field partition function
\begin{equation}
\mathcal{Z}_{Matter} = (\mathcal{Z}_{GFF})^{c_m}
\end{equation}
although other choices are also possible.

\subsection{Computation of $\mathcal{Z}_{GFF}$ and $\mathcal{Z}_{Ghost}$ for the annulus}\label{sec_ghost}
We now explain how to obtain the expressions of $\mathcal{Z}_{GFF}$ and $\mathcal{Z}_{Ghost}$ used in section 4. It turns out it is more convenient to compute these quantities if we represent our domain by a cylinder of unit circumference and of length $l$. In the complex plane our cylinder is a rectangle of sides of length $1$ and $l$ where we identify both sides of length $l$. We will then use the map $z \mapsto e^{-2 \pi i z}$ which maps this cylinder to our annulus with radii $1$ and $\tau$ given that we have the relation $\tau =  e^{2 \pi l}$. We will first compute $\mathcal{Z}_{GFF}$ on the cylinder, the eigenvalues of $- \Delta$ are given by
\begin{equation}
\lambda_{m,n} = (2 \pi m)^2 + (\frac{\pi n}{l})^2
\end{equation}
with $m \in \mathbb{Z}$ and $n \in \mathbb{N}$ for Neumann boundary conditions. The expression of $\mathcal{Z}_{GFF}$ is given in terms of the formal product $ \det{}' \Delta = \prod_{m,n}' \lambda_{m,n} = + \infty $ where the prime indicates that we omit the zero eigenvalue. The standard technique to give a finite value to this quantity is to use the zeta function regularization method. Similar computations can found in \cite{Sarnak}, \cite{Weis} and for a more general discussion of the method one can consult \cite{Bilal}. We  introduce:
\begin{equation}
\zeta(s) = \sum_{m,n}{}^{'} \frac{1}{\lambda_{m,n}^s}.
\end{equation}
This function is well defined for $ \Re(s)$ sufficiently large. We can then perform an analytic continuation to the complex plane and use the following relation:
\begin{equation}
 \det{}' \Delta = \exp \sum_{m,n}{}^{'} \ln \lambda_{m,n} = \exp - \frac{\partial}{\partial s} \zeta(s)\vert_{s=0}.
\end{equation}
With this method we obtain
\begin{equation}
 \det{}' \Delta = 2 l \eta(e^{2 \pi l })^2 
\end{equation}
with
\begin{equation}
\eta(x) = x^{-\frac{1}{12}} \prod_{n=1}^{\infty}(1-x^{-2n}).
\end{equation}
Let $A =l$ be the area of our cylinder. The partition function of the GFF is then given by \cite{Huber}:
\begin{equation}
\mathcal{Z}_{GFF}(l,dx^2) = \sqrt{A} \det{}' \Delta^{-1/2} = \frac{1}{\sqrt{2}} \frac{1}{\eta(e^{2\pi l})}.
\end{equation}
We must now find the link between $\mathcal{Z}_{GFF}(l,dx^2)$ and $\mathcal{Z}_{GFF}(\tau,dx^2)$. For this we will use the conformal anomaly formula \eqref{weyl_boun}. The cylinder equipped with the flat metric corresponds to the annulus with the metric $e^{\varphi} dx^2$ with $ \varphi = -2 \ln( 2 \pi \vert z \vert)$. To see this we introduce the coordinates $z =x +iy$ on the cylinder and  the coordinates $z' = r e^{i \theta}$ on the annulus. We have the relations $r = e^{2 \pi y} $, $ \theta = - 2 \pi x$ and the  relation between the metrics $dz'^2 = \vert \psi'(z) \vert^2 dz^2$ with $\vert \psi'(z) \vert^2 = 4 \pi^2 \vert e^{-2 \pi i z} \vert^2 = 4\pi^2 e^{4\pi y} = 4 \pi^2 r^2$. Therefore the cylinder with the flat metric corresponds to the annulus with the metric $\frac{1}{4 \pi^2 r^2} dz'^2$. To use the conformal anomaly formula we compute for $\varphi(z) = -2 \ln( 2 \pi \vert z \vert)$,
\begin{align*}
&\int_{\Omega} \vert \partial \varphi \vert^2 d \lambda = \int_{\theta = 0}^{2 \pi} \int_{r=1}^{\tau} \frac{4}{r^2} r dr d\theta = 8 \pi \ln \tau, \\
&\int_{\partial \Omega} 4 K \varphi d \lambda_{\partial} = -16 \pi \ln \tau.
\end{align*}
Therefore we get:
\begin{align*}
\mathcal{Z}_{GFF}(l, dx^2) = \mathcal{Z}_{GFF} ( \tau, e^{\varphi} dx^2) &= \exp \left( \frac{1}{96 \pi} \int_{\Omega} \vert \partial \varphi \vert^2 d \lambda  + \frac{1}{96 \pi}\int_{\partial \Omega} 4 K \varphi d \lambda_{\partial} \right) \mathcal{Z}_{GFF}(\tau, dx^2) \\
&= \tau^{-1/12} \mathcal{Z}_{GFF}(\tau, dx^2)
\end{align*}
which gives
\begin{equation}
\mathcal{Z}_{GFF}(\tau, dx^2) =  \tau^{1/12} \frac{1}{\eta(\tau)}.
\end{equation}
We have drop the irrelevant numerical factor $\frac{1}{\sqrt{2}} $. Now moving to the ghost partition function, we will start from the known result given in \cite{Martinec}:
\begin{equation}
\mathcal{Z}_{Ghost}(l, dx^2) = \vert \eta(e^{2 \pi l}) \vert^2.
\end{equation}
The ghost partition function has central charge $-26$ so the conformal anomaly gives a factor $\tau^{-\frac{13}{6}}$. In this case we also have to include a factor coming from the change of variable on the integration measure, namely:
\begin{equation}
 Dl = \frac{1}{2 \pi \tau } D \tau.
\end{equation}
Putting everything together we arrive at:
\begin{equation}
\mathcal{Z}_{Ghost}(\tau, dx^2) = \frac{1}{\tau} \tau^{-\frac{13}{6}}  \vert \eta(\tau) \vert^2.
\end{equation}

\subsection{Green's function on the annulus}\label{sec_green}
    In this section we give the detailed computation of the Green's function. Similar computations can be found in \cite{Mel}. We will work on the annulus of radii $a$ and $b$ with $a<b$. The case of interest to us is then obtained by choosing $a =1$ and $b= \tau$. We want to find the Green's function for the following problem
   $$ \left\{  \begin{array}{lcl} \Delta u(z) = - 2 \pi f(z) & \text{for} & z\in \Omega  \\ \frac{d u(r,\theta) }{dr} |_{r=a}  = \alpha  \int_{\Omega} f(r,\theta) r dr d\theta & \text{for} &  z\in \partial \Omega_a  \\ \frac{d u(r,\theta) }{dr} |_{r=b}  = \beta  \int_{\Omega} f(r,\theta) r dr d\theta & \text{for} &  z\in \partial \Omega_b \\  \int_{\partial \Omega} u d \lambda_{\partial} = 0 \end{array} \right. $$
    meaning that we want to have for all functions $f$
    $$ u(r,\theta) = \int_{\Omega} G(r,\theta, \rho, \phi) f(\rho, \phi) \rho d \rho d\phi. $$
Here $\alpha$ and $\beta$ are two real constants, we will discuss their values latter on. We work in polar coordinates $z = r e^{i \theta}$ and write the Fourier decompositions of $u$, $f$, and $G$:
\begin{align*}
&u(r,\theta) = \frac{u_0(r)}{2} +\sum_{n=1}^{\infty} (u_n^c(r)\cos(n \theta) + u_n^s(r) \sin(n \theta) ) \\
&f(r,\theta) = \frac{f_0(r)}{2} +\sum_{n=1}^{\infty} (f_n^c(r)\cos(n \theta) + f_n^s(r) \sin(n \theta) ) \\
&G(r, \theta, \rho, \phi) = g_0(r, \rho) + 2 \sum_{n=1}^{\infty} \tilde{g}_n(r, \rho) \cos n (\theta - \phi)
\end{align*}
These three functions are then linked by the relations:
\begin{align*}
&u_0(r) = 2 \pi \int g_0(r,\rho) f_0(\rho) \rho d \rho \\
&u_n^c(r) = 2 \pi \int \tilde{g}_n(r, \rho) f_n^c(\rho) \rho d \rho \\
&u_n^s(r) = 2 \pi \int \tilde{g}_n(r, \rho) f_n^s(\rho) \rho d \rho
\end{align*}
The Laplace equation reads for both the sine and cosine parts of $u_n$ and $f_n$ and for $ n \in \mathbb{N}$:
    $$ \frac{d}{dr}(r \frac{d u_n(r)}{dr}) - \frac{n^2}{r} u_n(r) = -2 \pi r f_n(r).$$
The general solutions of this equation read
\begin{align*}
&u_0(r) = 2 \pi \int_a^r \ln(x/r)f_0(x)x dx + C_0 \ln(r) + D_0 \\
&u_n(r) = \frac{\pi}{n} \int_a^r ((\frac{x}{r})^n - (\frac{r}{x})^n) f_n(x) x dx + C_n r^n + D_n r^{-n}, n \neq 0
\end{align*}
where $C_n$ and $D_n$ are constants that will be fixed by the boundary conditions. For $n \geq 1$, our boundary conditions require $\frac{d u_n(r)}{dr}$ to vanish on both boundaries. We compute:
    $$ \frac{d u_n(r)}{dr} = -\frac{\pi}{r} \int_a^r ( (  \frac{x}{r})^n + (  \frac{r}{x})^n ) f_n(x) x dx + C_n n r^{n-1} - D_n n r^{-n-1}. $$ 
    The constants $C_n$ and $D_n$ are then determined for $n>0$ by the equations:
    \begin{align*}
    &C_n n a^{n-1} - D_n n a^{-n-1} = 0 \\
     &\frac{\pi}{b} \int_a^b ( (  \frac{x}{b})^n + (  \frac{b}{x})^n ) f_n(x) x dx - C_n n b^{n-1} + D_n n b^{-n-1} = 0
    \end{align*}
We get
\begin{align*}
 C_n =  \frac{\pi b^{n}}{  n(  b^{2n} - a^{2n})} \int_a^b ( (\frac{x}{b})^n  + (\frac{b}{x})^n) f_n(x) x  dx \\
 D_n = \frac{\pi b^{n} a^{2n}}{  n(  b^{2n} - a^{2n})} \int_a^b ( (\frac{x}{b})^n  + (\frac{b}{x})^n) f_n(x) x  dx \\
\end{align*}
and the expression of $u_n$
\begin{small}
\begin{equation*}
u_n(r) = \frac{\pi}{n} \int_a^r ((\frac{x}{r})^n - (\frac{r}{x})^n) f_n(x) x dx + \frac{\pi b^{n} r^n}{n (b^{2n} - a^{2n} )} \int_a^b ((\frac{x}{b})^n + (\frac{b}{x})^n) f_n(x) x dx + \frac{\pi b^{n} a^{2n} r^{-n}}{n (b^{2n} - a^{2n} )} \int_a^b ((\frac{x}{b})^n + (\frac{b}{x})^n) f_n(x) x dx.
\end{equation*}
\end{small}
This leads to a contribution to the Green's function given by:
    $$ \tilde{g}_n(r,\rho) = \frac{r^{-n} \rho^{-n}}{2n (b^{2n} - a^{2n})} \left \{ \begin{array}{lcl} (b^{2n} + \rho^{2n})(r^{2n} + a^{2n}), & \text{for} & r \leq \rho \\ (b^{2n} + r^{2n})(\rho^{2n} + a^{2n}), & \text{for} &  r \geq \rho \end{array} \right. $$
    
    We will rewrite the Green's function to have one uniformly convergent series and one term corresponding to the divergence when $z = z'$. In order to perform this we write for $r < \rho$:
    $$ \tilde{g}_n(r,\rho) = \frac{r^{-n} \rho^{-n}}{2n (b^{2n} - a^{2n})} (b^{2n} + \rho^{2n})(r^{2n} + a^{2n})
    = \frac{r^{-n} \rho^{-n}}{2n } ( \frac{a^{2n}}{b^{2n}(b^{2n} - a^{2n})}   + \frac{1}{b^{2n}} ) (b^{2n} + \rho^{2n})(r^{2n} + a^{2n})$$
    We will introduce:
    $$ g_n(r,\rho) = \frac{r^{-n} \rho^{-n} a^{2n}}{2n b^{2n}(b^{2n} - a^{2n})} \left \{ \begin{array}{lcl} (b^{2n} + \rho^{2n})(r^{2n} + a^{2n}), & \text{for}& r \leq \rho \\ (b^{2n} + r^{2n})(\rho^{2n} + a^{2n}), & \text{for} &  r \geq \rho \end{array} \right. $$
    We see that the Green's function can be written in terms of an absolutely convergent sum of $ g_n$ plus a term we can compute explicitly using the following formula:
    $$ \sum_{n=1}^{\infty} \frac{x^n}{n} \cos n\theta =  - \frac{1}{2} \ln (1 - 2 x \cos \theta + x^2). $$
    The additional term is given by:
    $$ \sum_{n=1}^{\infty} \frac{1}{n b^{2n} r^{n} \rho^{n}}  (b^{2n} + \rho^{2n})(r^{2n} + a^{2n}) \cos n (\theta - \phi)
     = \ln \frac{|b^4 z^2 z'^2|}{|a^2 - z \overline{z'}| |b^2 - z \overline{z'}||z - z'| |b^2 z - a^2 z'|}. $$
    For $r > \rho$, the expression is obtained by exchanging $z$ and $z'$. This just changes the last term in the $\ln$ from $|b^2 z - a^2 z'|$ to $|a^2 z - b^2 z'|$.

    Now looking at the condition for $n=0$, we have
    $$ \frac{d u_0(r)}{dr} = -\frac{2 \pi}{r} \int_a^r f_0(x) x dx + \frac{C_0}{r} $$
    and the boundary conditions give:
    \begin{align*}
    &\frac{C_0}{a} = 2  \pi \alpha \int_a^b f_0(x) x dx \\
    &\frac{C_0}{b}  -\frac{2 \pi}{b} \int_a^b f_0(x) x dx =  
    2 \pi \beta \int_a^b f_0(x) x dx
    \end{align*}
    Therefore, we see that for there to be a solution, we cannot pick the constants $\alpha$ and $\beta$ arbitrarily: they are linked by the relation $\alpha a - \beta b =1$.
    We can check that this relation is also a consequence of the Green-Riemann formula.
   The last boundary condition leads to $2\pi a u_0(a) + 2 \pi b u_0(b)  = 0$ which gives:
    $$ C_0(a \ln a + b \ln b) + D_0(a + b) + 2 \pi b \int_a^b \ln(x/b) f_0(x) x dx = 0.$$
    Therefore:
    $$ D_0 = - \frac{a \ln a + b \ln b}{a + b} C_0 - \frac{2 \pi b}{a+b} \int_a^b \ln(x/b) f_0(x) x dx. $$
We have one constant left, $C_0$, we will choose its value to make the Green's function symmetric. The relations $C_0 = 2 \pi a \alpha \int f_0(x) x dx $ and $ \alpha a - \beta b =1$ then completely determine all the constants in the problem. We have for $r \leq \rho$
    $$ g_0(r, \rho) = \alpha a ( \frac{a \ln r/a + b \ln r/b}{a+b}) - \frac{b}{a+b} \ln \rho/b $$
    and for $r \geq \rho$
    $$ g_0(r, \rho) = \ln \rho/r +  \alpha a ( \frac{a \ln r/a + b \ln r/b}{a+b}) - \frac{b}{a+b} \ln \rho/b. $$
     The value of $\alpha$ that makes the Green's function symmetric is $\alpha = \frac{1}{a+b}$. We then have the following expression for $g_0$:
    $$ g_0(r,\rho) = \left\{ \begin{array}{lcl}  \frac{a^2 \ln(r/a) + b^2 \ln (b/\rho) + a b \ln (r/ \rho) }{(a+b)^2}, & \text{for} & r \leq \rho \\  \frac{a^2 \ln(\rho/a) + b^2 \ln (b/r) + a b \ln ( \rho/r) }{(a+b)^2}, & \text{for}&  r \geq \rho \end{array} \right.$$    
Therefore we get the final expression for our Green's function
    $$ G(r,\theta,\rho, \phi) =   g_0(r,\rho) +2\sum_{n=1}^{\infty} {g}_n(r,\rho) \cos n(\theta - \phi) + \ln \frac{|b^4 z^2 z'^2|}{|a^2 - z \overline{z'}| |b^2 - z \overline{z'}||z - z'| |b^2 z - a^2 z'|}  $$
    where again the factor $|b^2 z - a^2 z'|$ holds for $r < \rho$ and is replaced by $|a^2 z - b^2 z'|$ for $r > \rho$.

\end{document}